\documentclass[a4paper,11pt]{article}
\usepackage[utf8]{inputenc}
\pdfoutput=1 
\usepackage[margin=3cm]{geometry}

\usepackage{amsfonts}
\usepackage{amsmath}
\usepackage{amssymb}
\usepackage{amsthm}
\usepackage{float}
\usepackage{authblk}
\allowdisplaybreaks

\usepackage{microtype}
\usepackage{hyperref}
\usepackage[svgnames]{xcolor}
\hypersetup{colorlinks={true},linkcolor={DarkBlue},citecolor={DarkGreen}}

\usepackage[capitalise,,nameinlink]{cleveref}

\usepackage{tikz}
\usetikzlibrary{positioning,arrows.meta}
\usepackage{subcaption}

\theoremstyle{definition}
\newtheorem{definition}{Definition}

\theoremstyle{plain}
\newtheorem{theorem}{Theorem}
\newtheorem{lemma}[theorem]{Lemma}
\newtheorem{corollary}[theorem]{Corollary}
\newtheorem{proposition}[theorem]{Proposition}

\theoremstyle{remark}

\newtheorem*{remark*}{Remark}
\newtheorem{claim}{Claim}
\newtheorem*{claim*}{Claim}

\renewcommand*\qedsymbol{$\blacksquare$}
\newenvironment{claimproof}[1][\proofname]{\begin{proof}[#1]\renewcommand*{\qedsymbol}{$\square$}}{\end{proof}}

\DeclareMathOperator*{\argmin}{argmin}

\title{The Hairy Ball problem is PPAD-complete\thanks{A preliminary version of this paper appeared in the proceedings of the 46th International Colloquium on Automata, Languages, and Programming (ICALP 2019).}}
\author[1]{Paul W. Goldberg}
\author[2]{Alexandros Hollender}
\affil[1]{Department of Computer Science, University of Oxford \ \ \ \ \ \ \ \  
	{\tt Paul.Goldberg@cs.ox.ac.uk}}
\affil[2]{Department of Computer Science, University of Oxford \ \ \ \ \ \ \ \ {\tt alexandros.hollender@cs.ox.ac.uk}}

\begin{document}

\maketitle

\begin{abstract}
The Hairy Ball Theorem states that every continuous tangent vector field on an even-dimensional sphere must have a zero. We prove that the associated computational problem of (a) computing an approximate zero is PPAD-complete, and (b) computing an exact zero is FIXP-hard. We also consider the Hairy Ball Theorem on toroidal instead of spherical domains and show that the approximate problem remains PPAD-complete. On a conceptual level, our PPAD-membership results are particularly interesting, because they heavily rely on the investigation of multiple-source variants of \textsc{End-of-Line}, the canonical PPAD-complete problem. Our results on these new \textsc{End-of-Line} variants are of independent interest and provide new tools for showing membership in PPAD. In particular, we use them to provide the first full proof of PPAD-completeness for the \textsc{Imbalance} problem defined by Beame et al.\ in 1998.
\end{abstract}


\newpage

\tableofcontents

\newpage

\section{Introduction}

The Hairy Ball Theorem (HBT) is a well-known topological theorem stating that there is no non-vanishing continuous tangent vector field on an even-dimensional $k$-sphere. It has various informal statements such as ``you can't comb a hairy ball flat without creating a cowlick''\footnote{\url{https://en.wikipedia.org/wiki/Hairy_ball_theorem}}, or ``there is a point on the surface of the earth with zero horizontal wind velocity''. The HBT is superficially reminiscent of the Borsuk-Ulam Theorem, stating that given any continuous mapping from the 2-sphere to the plane, there are two antipodal points that map to the same value. (Informally, ``there are two antipodal points on the surface of the earth where the temperature and pressure are the same''). As we shall see, the present paper highlights a fundamental difference between the two, in terms of the complexity class naturally associated with each of them.

The HBT was first proved in 1885 by Poincar{\'e}~\cite{poincare1885courbes} for the case $k=2$. The theorem as stated for all even $k$ was proved in 1912 by Brouwer~\cite{brouwer1912abbildung}. Accordingly, this result is sometimes also called the Poincar{\'e}-Brouwer theorem.
In fact, the result proved by Poincar{\'e}~\cite{poincare1885courbes} is stronger than stated above. It follows from it that for any (sufficiently well-behaved) 2-dimensional manifold with genus $g \neq 1$, any continuous tangent vector field must have a zero. In particular, this means that the HBT also holds for the torus of genus $g$ for $g \geq 2$, i.e.\ the 2-dimensional torus with $g$ holes. It is easy to see that it does not hold for the standard single-hole torus.

Over the years, various papers in the {\em American Mathematical Monthly} have presented alternative proofs of the Hairy Ball Theorem and variants, for example \cite{jarvis2004hairy,milnor1978hairyball,boothby1971hairyball,eisenberg1979hairyball,mcgrath2016hairyball,curtin2018hairyball}.

\smallskip
Topological existence results (such as the HBT, Borsuk-Ulam, and the Brouwer and Banach fixpoint theorems) have a very interesting relationship with complexity classes of search problems in which any instance has a guaranteed solution. Any such theorem has a corresponding computational challenge, of searching for such a solution, given a circuit that computes an appropriate function. The assumption that these complexity classes are distinct from each other (the ones of main interest here being PPAD and PPA, discussed below in more detail) provides a taxonomy of these theorems.
Our results highlight a fundamental distinction between the HBT and Borsuk-Ulam, by showing that the corresponding search problem for the HBT is characterised by the complexity class PPAD, in contrast to Borsuk-Ulam, which is characterised by PPA~\cite{aisenberg20152}. The complexity-theoretic analysis of topological search problems provides a well-defined sense in which the HBT is ``Brouwer-like'' rather than ``Borsuk-Ulam-like''. It has previously been noted that the HBT may be used to prove Brouwer's fixed point theorem~\cite{milnor1978hairyball}, but not the other way around. Indeed, the existing proof of HBT using Sperner's Lemma~\cite{jarvis2004hairy} actually uses a generalisation of Sperner's Lemma, which was not known to be equivalent to Brouwer's fixed point theorem prior to our work.

\subsection{Background on NP total search and PPAD}

The complexity class TFNP is the set of all {\em total} function computation problems in NP: functions where {\em every} input has an efficiently-checkable solution (in \cref{sec:tfnp} we give a precise definition). Many problems in TFNP appear to be computationally difficult, notably \textsc{Factoring}, the problem of computing a prime factorisation of a given number, also \textsc{Nash}, the problem of computing a Nash equilibrium of a game. However, such problems are unlikely to be NP-hard, due to the 1991 result of Megiddo and Papadimitriou~\cite{megiddo1991total} showing that TFNP problems cannot be NP-hard unless NP is equal to co-NP. This basic fact, that hard TFNP problems are in a very strong sense ``NP-intermediate'', provides TFNP's strong theoretical appeal. This has led to the classification of these problems in terms of certain syntactic subclasses of TFNP, whose problems are shown to be total due to some basic combinatorial principle. The best-known of these classes are PLS, PPP, PPAD, and PPA, identified by Papadimitriou in 1994~\cite{papadimitriou1994complexity}.
\begin{itemize}
\item PPAD consists of problems whose totality is based on the principle that given a source in a directed graph whose vertices have in-degree and out-degree at most 1, there exists another degree-1 vertex. Its canonical problem \textsc{End-of-Line} consists of an exponentially-large graph of this kind, presented concisely via a circuit.
\item PPA differs from PPAD in that the graph need not be directed; being a more general principle, PPA is thus a superset of PPAD. Its canonical problem \textsc{Leaf} is similar, only the graph is undirected.
\end{itemize}
Subsequently, many TFNP problems of interest were shown PPAD-complete~\cite{daskalakis2009complexity,chen2009complexity,KPRST13,DQS12cake}, while more recently others were shown PPA-complete~\cite{FRGstoc2018,FRGstoc2019}. Despite their similar definitions, PPAD and PPA are usually conjectured to be different, and (along with other syntactic TFNP subclasses) are separated by oracles~\cite{beame1998relative}.

\subsection{Our results and their significance}

Given the long-standing interest in the Hairy Ball Theorem, it is natural to study the corresponding computational search problem. In this paper, we prove that computing an approximate zero of a Hairy Ball vector field is PPAD-complete. While many PPAD-completeness results already exist, a noteworthy novelty of our results is that we find that computing HBT solutions corresponds with {\em multiple-source} variants of the \textsc{End-of-Line} problem: given a large directed graph implicitly represented by a circuit, suppose you are shown several sources and told to find another degree-1 vertex. This is in contrast with previous PPAD-complete problems that naturally reduce to standard single-source \textsc{End-of-Line}.

In \cref{sec:EOL} we prove that these multiple-source \textsc{End-of-Line} variants are PPAD-complete (membership of PPAD being the tricky aspect). Our results make progress on the general question (studied in~\cite{goldberg2018towards}) of whether there exist combinatorial principles indicating totality of search problems, that are fundamentally different from the known ones that give rise to complexity classes such as PPAD. In particular, in \cref{sec:imbalance}, we note that a proof of PPAD-completeness for the \textsc{Imbalance} problem by Beame et al.~\cite{beame1998relative} is incomplete and provide a full proof using our results.

The generalisation of Poincar{\'e}'s result to higher dimensions is called the Poincar{\'e}-Hopf theorem (see e.g.~\cite{guillemin1974differential}). This theorem relates the number and types of zeros of a vector field on a manifold with its Euler characteristic, a topological invariant. In particular, if the Euler characteristic of a manifold is not $0$, then any continuous tangent vector field on the surface must have a zero. The Euler characteristic of even-dimensional spheres is $2$, while it is $2(1-g)$ for 2-dimensional toruses of genus $g \geq 2$. For odd-dimensional spheres it is $0$.

We believe that the reduction to multiple-source \textsc{End-of-Line} is not an artefact of our techniques, but instead intrinsically related to the Euler characteristic of the domain. Indeed, the reduction from the HBT problem on even-dimensional spheres to \textsc{End-of-Line} yields $2$ sources (\cref{sec:HBinPPAD}). On the other hand, if we consider the HBT problem on the 2-dimensional torus of genus $g \geq 2$, then we obtain $2(g-1)$ sources (\cref{sec:gTorus}). The connection between HBT and directed graph problems has previously only appeared in a proof for the 2-dimensional sphere case~\cite{jarvis2004hairy}.

Finally, we note that PPAD-hardness is obtained by constructing a HBT vector field from multiple copies of a discrete Brouwer fixpoint problem. The usage of multiple copies is a new conceptual feature, closely related to the multi-source aspect. Using the same high-level idea, we also provide a FIXP-hardness result for the problem of computing an \emph{exact} solution (\cref{sec:HBFIXPhard}).

\subsection{Other related work}

Banach's Fixed Point Theorem~\cite{banach1922} says that a {\em contraction map} has a unique fixpoint. Its corresponding computational problem {\sc Contraction}, is to find a fixed point of a given contraction map. Some versions of {\sc Contraction} have been shown complete for CLS, a subclass\footnote{A very recent result~\cite{FearnleyGHS-gradient} shows that, in fact, CLS = PPAD $\cap$ PLS.} of PPAD \cite{daskalakis2011continuous,daskalakis2018converse,fearnley2017cls}.
The search for Brouwer fixpoints (including discretised versions of Brouwer functions) is PPAD-complete for most variants of the problem~\cite{papadimitriou1994complexity,chen2009complexity}, which is why we say the HBT is ``Brouwer-like''.
Finally --- in contrast --- the computational problem of searching for a Borsuk-Ulam solution is PPA-complete~\cite{aisenberg20152}. Other topological existence results that have PPA-complete search problems include the Hobby-Rice theorem~\cite{FRGstoc2018} and the Ham Sandwich Theorem~\cite{FRGstoc2019}.

\subsection{Future research directions}

We have obtained a satisfying answer to the question of the computational complexity of the Hairy Ball Theorem, if we are looking for an approximate solution. For other solution concepts related to exact solutions, we have provided a FIXP-hardness result. This leaves open the question of whether the problem is FIXP-complete in this case. A first step in that direction would be to try to reduce Hairy Ball to Borsuk-Ulam, even though no such (fully constructive) mathematical proof seems to be known.

Our reduction from approximate Hairy Ball to approximate Brouwer cannot be used on the exact versions of these problems, because the step where we reduce from multi-source \textsc{End-of-Line} to standard \textsc{End-of-Line} corresponds to a \emph{discontinuous} mapping between the corresponding topological structures. This raises the question of whether continuous mappings are an important subclass of reductions and motivates further study on this topic.

Our results on multiple-source variants of the \textsc{End-of-Line} problem open the way for two new research directions. First, they provide a new tool for showing membership of PPAD, which can be used to put further problems in this class. It seems very unlikely that the Hairy Ball Theorem should be the only ``natural'' application of these results. Furthermore, a second interesting research direction is investigating the complexity of \textsc{End-of-Line} with a super-polynomial number of known sources (implicitly given in the input).

\section{Preliminaries}

Let $k$ be a positive integer. For $x \in \mathbb{R}^k$, $\|x\|_2$, $\|x\|_1$ and $\|x\|_\infty$ denote the standard $\ell_2$-norm, $\ell_1$-norm and $\ell_\infty$-norm respectively. For $x,y \in \mathbb{R}^k$, $\langle x,y \rangle := \sum_{i=1}^k x_i y_i$ denotes the inner product.

The $k$-dimensional unit sphere in $\mathbb{R}^{k+1}$ (or $k$-sphere) is denoted $S^k = \{x \in \mathbb{R}^{k+1} : \|x\|_2=1\}$. A continuous tangent vector field on $S^k$ is a continuous function $f: S^k \to \mathbb{R}^{k+1}$ such that for all $x \in S^k$ we have $\langle f(x), x \rangle = 0$. The Hairy Ball Theorem can be stated as follows:

\begin{theorem}[Poincar{\'e}~\cite{poincare1885courbes}--Brouwer~\cite{brouwer1912abbildung}]
If $k \geq 2$ is even, then for any continuous tangent vector field $f: S^k \to \mathbb{R}^{k+1}$, there exists $x \in S^k$ such that $f(x)=0$.
\end{theorem}

\subsection{Model of Computation}

We work in the standard Turing machine model. All numbers appearing in computations are rational numbers where the numerator and denominator are integers represented in binary. For a rational number $x$, $\text{size}(x)$ denotes the size of the representation of $x$, i.e.\ the sum of the representation length of its numerator and denominator in binary. For an arithmetic circuit $F$, $\text{size}(F)$ denotes the number of gates in the circuit added to the representation length of any rational constants used by the circuit.

\subsection{Formal definition of TFNP}\label{sec:tfnp}

A computational search problem is given by a binary relation $R \subseteq \{0,1\}^* \times \{0,1\}^*$, interpreted as follows: $y \in \{0,1\}^*$ is a solution to instance $x \in \{0,1\}^*$, if and only if $(x,y) \in R$. The search problem $R$ is in FNP (\emph{Functions in NP}), if $R$ is polynomial-time computable (i.e.\ $(x,y) \in R$ can be decided in polynomial time in $|x|+|y|$) and there exists some polynomial $p$ such that $(x,y) \in R \implies |y| \leq p(|x|)$. Here $\{0,1\}^*$ denotes all finite length bit-strings and $|x|$ is the length of bit-string $x$.

The class TFNP (\emph{Total Functions in NP}~\cite{megiddo1991total}) contains all search problems $R$ that are in FNP and are \emph{total}, i.e.\ every instance has at least one solution. Formally, this corresponds to requiring that for every $x \in \{0,1\}^*$ there exists $y \in \{0,1\}^*$ such that $(x,y) \in R$.

Let $R$ and $S$ be total search problems in TFNP. We say that $R$ (many-one) reduces to $S$, if there exist polynomial-time computable functions $f,g$ such that
$$(f(x),y) \in S \implies (x,g(x,y)) \in R.$$
Note that if $S$ is polynomial-time solvable, then so is $R$. We say that two problems $R$ and $S$ are (polynomial-time) equivalent, if $R$ reduces to $S$ and $S$ reduces to $R$.

\subsection{The End-of-Line problem}

To be PPAD-complete, a problem must be equivalent to \textsc{End-of-Line}. The \textsc{End-of-Line} problem is informally defined as follows: given a directed graph where each vertex has in- and out-degree at most $1$ and given a known source of this graph, find a sink or another source. The problem is computationally challenging, because the graph is not given explicitly in the input. Instead, we are given an implicit concise representation of the graph through circuits that compute the predecessor and successor  of a vertex in the graph. In what follows, we sometimes interpret the input and output of the circuits, which are elements in $\{0,1\}^n$, as the numbers $\{0,1, \dots, 2^n-1\}$.

\begin{definition}[\textsc{End-of-Line}~\cite{daskalakis2009complexity}]\label{def:endofline}
The \textsc{End-of-Line} problem is defined as: given Boolean circuits $S,P$ with $n$ input bits and $n$ output bits and such that $P(0) = 0 \neq S(0)$, find $x$ such that $P(S(x)) \neq x$ or $S(P(x)) \neq x \neq 0$.
\end{definition}
The circuits define a graph as follows. There is a directed edge from vertex $x$ to $y$ ($x \neq y$), if and only if $S(x)=y$ and $P(y)=x$. Note that any badly defined edge, i.e. $S(x) = y$ and $P(y) \neq x$, or $P(y) = x$ and $S(x) \neq y$, qualifies as a solution of \textsc{End-of-Line} as defined above (because $P(S(x)) \neq x$ or $S(P(x)) \neq x$ respectively). Note that $0$ is a source of the graph, unless $P(S(0)) \neq 0$, in which case $0$ is a valid solution to the problem as stated above.

It is easy to check that this formal definition of the problem is computationally equivalent to the informal description given above. By definition, \textsc{End-of-Line} is PPAD-complete~\cite{papadimitriou1994complexity}. Furthermore, reduction from \textsc{End-of-Line} is a very common technique to show PPAD-hardness (e.g.~\cite{daskalakis2009complexity,chen2009complexity}). In \cref{sec:EOL} we show that the multiple-source version \textsc{MS-EoL} (\cref{def:mulsourceeol}) is equivalent to \textsc{End-of-Line}.

\section{The Hairy-Ball Problem}

\subsection{The \texorpdfstring{$k$}{k}D-Hairy-Ball problem}

The Hairy Ball Theorem naturally yields a corresponding computational problem. We are given a continuous tangent vector field $f$ on the unit sphere and have to find a point where it is zero. In trying to formalise this, some issues need to be addressed. First, one has to decide how the vector field should be represented in the input. Here we take the usual approach of assuming that it is represented as an arithmetic circuit.

Before we discuss the types of gates that we want to allow in the circuit, let us briefly handle the second issue: the vector field might not have a rational zero. Indeed, consider the following example: at $x \in S^2$ the vector field is simply the vector $(1,1,1)$ projected onto the tangent space of $S^2$ at $x$. In this case, the only solutions are $\pm (1/\sqrt{3},1/\sqrt{3},1/\sqrt{3})$. Thus, we cannot hope to always output an exact solution. We bypass this problem by asking for an \emph{approximate} solution instead, i.e.\ a point $x \in S^2$ such that $\|f(x)\|_\infty \leq \varepsilon$ for some $\varepsilon > 0$ provided in the input. This notion of approximate solution is the standard one used when studying topological existence theorems in the context of TFNP (e.g.\ Brouwer's fixed point theorem or the Borsuk-Ulam theorem).

As mentioned above, the vector field will be represented as an arithmetic circuit. In the case of $S^2$, the circuit will have three input gates and three output gates. The arithmetic circuit will be allowed to use gates $\{+,-,\times \zeta, \max, \min\}$ and rational constants. All the gates have fan-in 2, except $\times \zeta$ which has fan-in 1 and corresponds to multiplication by a rational constant $\zeta$. Note that such a circuit is polynomially equivalent to a circuit only using gates $\{+, \times \zeta, \max\}$ and rational constants, since the other gates can be efficiently simulated using these. These circuits correspond to LINEAR-FIXP-type circuits that are known to be sufficient to obtain PPAD-hardness of \textsc{Brouwer}~\cite{etessami2010fixp}. A discussion about why we don't use more powerful gates in our definition can be found in the next section.

This type of circuit yields piece-wise affine functions that are continuous. Furthermore, it has the following nice property: for any such arithmetic circuit $F$, and any rational $x$, we can compute $F(x)$ exactly in polynomial time in $\text{size}(F)$ and $\text{size}(x)$. One potential issue is that $F$ might not be \emph{tangent} to the sphere, but this is easy to fix by simply considering the vector field given by the projection of $F$ onto the corresponding tangent space to the sphere. Thus, we define the computational problem as follows:

\begin{definition}[\textsc{$k$D-Hairy-Ball}]\label{def:linearHB}
Let $k \geq 2$ be even. The \textsc{$k$D-Hairy-Ball} problem is defined as: given $\varepsilon > 0$ and an arithmetic circuit $F$ with $k+1$ inputs and outputs, using gates $\{+, \times \zeta, \max \}$ and rational constants, find $x \in S^k$ such that $\left\| P_x [F(x)] \right\|_\infty \leq \varepsilon$.
\end{definition}

Here $P_x[ \cdot ]$ denotes the projection onto the tangent space to the sphere $S^k$ at $x \in S^k$. Note that for any $v \in \mathbb{R}^{k+1}$, we have $P_x[v] = v - \langle v,x \rangle x$, because $\|x\|_2 = 1$. Thus, the projection of any rational vector $v$ onto the tangent space at rational $x \in S^k$ can be computed exactly in polynomial time in $\text{size}(v)$ and $\text{size}(x)$. Note that we are looking for a solution with respect to the $\ell_\infty$-norm, but we could also have used the $\ell_2$- or $\ell_1$-norm, since all these versions are computationally equivalent.

\textsc{$k$D-Hairy-Ball} lies in TFNP. Clearly, any solution can be checked in polynomial time. Totality of \textsc{$k$D-Hairy-Ball} will immediately follow when we prove that it lies in PPAD (\cref{cor:linearHBinPPAD}).

\begin{lemma}\label{lem:linearHBlipschitz}
Let $k \geq 2$ be even. Let $F$ be an arithmetic circuit with $k+1$ inputs and outputs, using gates $\{+, \times \zeta, \max \}$ and rational constants. Then, the function $S^k \to \mathbb{R}^{k+1}$, $x \mapsto P_x[F(x)]$ is Lipschitz-continuous with Lipschitz constant $L = k \cdot 2^{\textup{size}(F)^2+3}$ (w.r.t. $\ell_\infty$-norm).
\end{lemma}

\begin{proof}
It is easy to see that if functions $f_1$ and $f_2$ are $L_1$- and $L_2$-Lipschitz respectively, then $f_1+f_2$ is $(L_1+L_2)$-Lipschitz and $\max\{f_1,f_2\}$ is $\max\{L_1,L_2\}$-Lipschitz. Since any rational constant used in $F$ has absolute value bounded by $2^{\text{size}(F)}$, it follows that $f_1 \times \zeta$ is $L_1 2^{\text{size}(F)}$-Lipschitz. Using the fact that each input gate corresponds to a $1$-Lipschitz function, it follows that $F$ must be $2^{\text{size}(F)^2}$-Lipschitz.

Using the same kind of argument it is easy to show that for $x \in S^k$ we always have $\|F(x)\|_\infty \leq 2^{\text{size}(F)^2}$. Using the fact that $f_1 \times f_2$ is $(\|f_1\|_\infty L_2 + L_1 \|f_2\|_\infty)$-Lipschitz and the definition $P_x[v] = v - \langle v,x \rangle x$, we obtain that $x \mapsto P_x[F(x)]$ is Lipschitz-continuous on $S^k$ with constant $L \leq (3k+4)2^{\text{size}(F)^2} \leq k \cdot 2^{\text{size}(F)^2+3}$.
\end{proof}

Our main result is \cref{thm:HBPPADcomplete}. Containment in PPAD, which turns out to be the most challenging part of this result, is presented in \cref{sec:HBinPPAD} (using the multiple-source \textsc{End-of-Line} results of \cref{sec:EOL}). PPAD-hardness is presented in \cref{sec:HBhardness}.
\begin{theorem}\label{thm:HBPPADcomplete}
For all even $k \geq 2$, \textsc{$k$D-Hairy-Ball} is \textup{PPAD}-complete.
\end{theorem}

\subsection{About the power of the arithmetic circuit}

The main disadvantage of defining \textsc{$k$D-Hairy-Ball} as in \cref{def:linearHB} is that the type of circuit used seems quite restrictive. Clearly, it would be natural to also allow $\times$-gates that compute the product of two intermediate outputs (as opposed to just multiplication by a fixed constant). Unfortunately, allowing the $\times$-gate leads to technical complications in the definition of the problem because of repeated squaring (see below). These technicalities would make \cref{def:linearHB} more cumbersome and less intuitive overall. In the interest of simplicity and clarity of exposition we have thus decided to use the simpler definition that only allows $\{+,\times \zeta, \max\}$ gates.

From the standpoint of computational complexity, this restriction turns out to be irrelevant. Indeed, in \cref{sec:HBinPPAD} we prove that a very abstract formulation of the Hairy-Ball problem lies in PPAD. Namely, we only require that the tangent vector field be polynomially computable and polynomially continuous (\cref{def:polycont-polycomp}). These two assumptions are very natural and desirable in any TFNP-style definition of the problem. In particular, properly defining the problem with the additional $\times$-gates would also yield a problem that lies in PPAD. Furthermore, in \cref{sec:HBhardness} we prove that \textsc{$k$D-Hairy-Ball} is PPAD-hard. It immediately follows that formulations with more powerful circuits are also PPAD-hard. Thus, it turns out that allowing $\times$-gates yields a problem that is polynomially equivalent to this more restricted version that we use.

Before we close this section, we give some details about the complication that arises when allowing the $\times$-gates and how this can be circumvented in order to define a TFNP-problem. The main issue is that if the circuit is allowed to use $\times$-gates, then we might not be able to evaluate the function it represents efficiently. For example, consider the arithmetic circuit that does repeated squaring, i.e.\ it has a sequence of $n$ $\times$-gates that multiply the output of the previous gate by itself. On input $2$, this circuit outputs $2^{2^n}$, which is doubly exponential in the size of the circuit and input to the circuit. Even putting aside the fact that we cannot efficiently represent this number in our model, this also causes the Lipschitz-constant to be doubly exponential. In order to solve this we can enforce an upper bound $M$ on all computations of the circuit. This means that every gate would have an output in $[-M,M]$. $M$ would have representation size polynomial in the size of the circuit. Without loss of generality we can take $M=1$. However, we still might not be able to compute the output of the (bounded) circuit exactly. Indeed, the same circuit as earlier, on input $1/2$, would output $2^{-2^n}$. We would need an exponential (w.r.t.~the number of bits of the input and the size of the circuit) number of bits to represent this.

As a result, we would have to settle for approximate computation of the circuit's output. Let $F$ be a circuit with computation bounded in $[-1,1]$. For any $m \in \mathbb{N}$ and any $x \in S^k$, we can compute $F_m(x)$, which is $F(x)$ up to error $2^{-m}$, in time polynomial in the size of the circuit, $m$ and $\text{size}(x)$. This can be achieved by computing the output of every gate up to some error. Thus, when defining the Hairy-Ball problem, a solution would be required to satisfy $\|F_m(x)\|_\infty \leq \varepsilon$ for $m=\lceil \log_2 (2/\varepsilon) \rceil$. Note that this implies that $\|F(x)\|_\infty \leq 3\varepsilon/2$.

One nice property of this kind of circuit is that the condition that $F$ be tangential to $S^k$ can be enforced syntactically. Indeed, we can extend the circuit to output $F(x) - \langle F(x),x \rangle x$, instead of $F(x)$. Some care is required here, because the computations of the circuit are bounded in $[-1,1]$, but this can be solved by multiplying $F(x)$ by $1/k$ before projecting and using $\varepsilon/k$ instead of $\varepsilon$. Thus, we could always assume that the circuit given in the input describes a tangential vector field on the sphere.

As mentioned above, we do not use this kind of circuit because it introduces unnecessary complications and clutter.\footnote{Another more elegant option would be to consider \emph{well-behaved} circuits, as defined in \cite{FearnleyGHS-gradient}.} However, a closer inspection of the proof of \cref{thm:HBinPPAD} yields that this problem also lies in PPAD and thus turns out to be equivalent to \textsc{$k$D-Hairy-Ball}.

\section{The Hairy-Ball Problem is in PPAD}\label{sec:HBinPPAD}

In this section we present our main result: the problem of computing an approximate Hairy Ball solution reduces to \textsc{End-of-Line}, the canonical PPAD-complete problem.

From a purely mathematical standpoint, our proof can be used to provide a (fairly cumbersome) proof of the Hairy Ball theorem by using Brouwer's fixed point theorem. Indeed, it is known~\cite{papadimitriou1994complexity} that \textsc{End-of-Line} reduces to \textsc{Brouwer} (in fact, even \textsc{2D-Brouwer}~\cite{chen2009complexity}). Thus, given a Hairy Ball function $f$, using our reduction and Brouwer's fixed point theorem, one can prove the existence of a point $x_k$ such that $\|f(x_k)\| \leq 1/2^k$ for any $k$ (using the fact that $f$ must be uniformly continuous since the sphere is compact). Then, since any sequence in a compact set must have a converging subsequence it follows that there must exist $x$ such that $f(x)=0$. Finding a more direct way to deduce the Hairy Ball theorem from Brouwer's fixed point theorem is an interesting open question.

\subsection{A general version of the Hairy-Ball problem}

Our goal is to prove that the Hairy Ball problem lies in PPAD in a setting that is as general and encompassing as possible. The way the function is represented, as a circuit or otherwise, should not play a role. Thus, we are only going to make two assumptions about the tangent vector field: that it can be evaluated in polynomial time and that it is polynomially continuous in some well-defined sense. The first assumption is very natural: if we are given a Hairy Ball function, we expect to be able to evaluate it efficiently. The motivation for the second assumption is that if we omit it, then there is no guarantee that there will exist an approximate solution with representation size that is polynomial in the input size.

We now define these assumptions formally, following the analogous definitions by Etessami and Yannakakis~\cite{etessami2010fixp} for Brouwer fixed point problems. Let $\mathcal{F}$ be a class of Hairy Ball functions $f : S^k \to \mathbb{R}^{k+1}$ (i.e.\ continuous tangent vector fields) with $k \geq 2$ even. Note that here $k$ is not fixed for all $f \in \mathcal{F}$, but we assume that $k \leq \text{size}(f)$. For any $f \in \mathcal{F}$, $\text{size}(f)$ denotes the length of the representation of $f$ in $\mathcal{F}$. In the case of \textsc{$k$D-Hairy-Ball}, $k$ is fixed and $\mathcal{F}$ is the class of all such functions represented using arithmetic circuits with gates $\{+, \times \zeta, \max\}$ (with the projection onto the tangent space at the end). In that case, $\text{size}(f)$ is the size of the circuit representing $f$. Recall that for rational vector $x$, $\textup{size}(x)$ is the length of the representation of $x$.

\begin{definition}[\cite{etessami2010fixp}]\label{def:polycont-polycomp}
Let $\mathcal{F}$ be a class of Hairy Ball functions.
\begin{itemize}
    \item $\mathcal{F}$ is \emph{polynomially computable}, if there exists some polynomial $p$ such that for any $f \in \mathcal{F}$ and any rational input $x \in S^k$, $f(x)$ can be computed in time $p(\text{size}(f) + \text{size}(x))$.
    
    \item $\mathcal{F}$ is \emph{polynomially continuous}, if there exists some polynomial $q$ such that for any $f \in \mathcal{F}$ and any rational $\varepsilon > 0$, there exists a rational $\delta > 0$ with $\text{size}(\delta) \leq q(\text{size}(f)+\text{size}(\varepsilon))$ such that for all $x,y \in S^k$ we have $\|x-y\|_\infty \leq \delta \implies \|f(x)-f(y)\|_\infty \leq \varepsilon$.
\end{itemize}
\end{definition}

Note that \textsc{$k$D-Hairy-Ball} yields a class $\mathcal{F}$ that is both polynomially computable and polynomially continuous (by \cref{lem:linearHBlipschitz}).

\begin{definition}
Let $\mathcal{F}$ be a class of Hairy Ball functions. The problem \textsc{Hairy-Ball($\mathcal{F}$)} is defined as: given $f \in \mathcal{F}$ and $\varepsilon > 0$, find $x \in S^k$ such that $\|f(x)\|_\infty \leq \varepsilon$.
\end{definition}

For simplicity we assume that we can recognise whether some string in $\{0,1\}^*$ represents an element $f \in \mathcal{F}$ in polynomial time. If this does not hold, then \textsc{Hairy-Ball($\mathcal{F}$)} has to be studied as a \emph{promise} problem. The reduction to \textsc{End-of-Line} given in the proof below still holds. However, this does not imply that the problem lies in PPAD, because TFNP requires the problem to be total without any promise.

\subsection{The general problem lies in PPAD}

\begin{theorem}\label{thm:HBinPPAD}
Let $\mathcal{F}$ be a class of Hairy Ball functions that is polynomially computable and polynomially continuous. Then, \textsc{Hairy-Ball($\mathcal{F}$)} lies in \textup{PPAD}.
\end{theorem}

\begin{corollary}\label{cor:linearHBinPPAD}
For all even $k \geq 2$, \textsc{$k$D-Hairy-Ball} lies in \textup{PPAD}.
\end{corollary}

\begin{proof}[Proof Overview for~\cref{thm:HBinPPAD}]\renewcommand{\qedsymbol}{}
The proof can be subdivided into two parts. In the first part, we reduce \textsc{$k$D-Hairy-Ball} to a \emph{2-source} \textsc{End-of-Line} problem. In the second part, we show that the 2-source version reduces to the standard version of \textsc{End-of-Line}, where a single source is known. Surprisingly, the reduction from 2 sources to 1 is non-trivial. The proof for this is presented separately in~\cref{sec:EOL}. In fact, we prove the more general result: as long as the number of known sources in an \textsc{End-of-Line} instance is polynomial, we can reduce to standard \textsc{End-of-Line}. Various implications of this result are also presented in~\cref{sec:EOL}.

We now give some details about the first part of the proof, in which we reduce \textsc{$k$D-Hairy-Ball} to 2-source \textsc{End-of-Line} through a Sperner argument. The inspiration for this comes from a proof of the 2-dimensional Hairy Ball Theorem via a generalisation of Sperner's Lemma, given by Jarvis and Tanton~\cite{jarvis2004hairy}. Our contribution here is two-fold: we extend their proof to any higher (even) dimension and we turn it into a polynomial-time reduction. Even though Sperner's Lemma is known to be PPAD-complete~\cite{papadimitriou1994complexity}, this result does not apply to the general version of Sperner's Lemma that is used here. Indeed, by using standard techniques, this general version can only be reduced to a \emph{2-source} instance of \textsc{End-of-Line}. This is why we then require our technical results about multiple-source \textsc{End-of-Line} from \cref{sec:EOL}.

In order to turn the ideas of Jarvis and Tanton into a polynomial-time reduction, instead of working directly on the sphere, we use a stereographic projection to ``unfold'' the sphere $S^k$ ($\subset \mathbb{R}^{k+1}$) into the space $\mathbb{R}^k$, along with the vector field. Then, we consider a sufficiently large cross-polytope $C$ of $\mathbb{R}^k$ and prove that the ``unfolded'' vector field satisfies certain boundary conditions. In the case $k=2$, this corresponds to the vector field making two full rotations when we move along the boundary of $C$ (see \cref{fig:HBinPPAD-2D-boundary}). Next, we pick an efficient triangulation of $C$ and a suitable colouring of its nodes. The last step then requires us to prove that this colouring yields exactly two starting points (on the boundary) for Sperner paths (see \cref{fig:HBinPPAD-2D-triangulated}) that lead to \emph{panchromatic} simplices. Using standard Sperner-arguments this yields a 2-source \textsc{End-of-Line} instance. The full proof for any even $k$ can be found below. Note that, as expected, the proof does not work if $k$ is odd. Indeed, the construction then yields a starting point and an ending point on the boundary, instead of two starting points.
\end{proof}

\begin{figure}[h]
\begin{subfigure}[t]{0.5\textwidth}
\centering
\begin{tikzpicture}
\fill[white] (-2.5,-2.5) rectangle ++ (5,5);
\fill[lightgray] (0,2) -- (2,0) -- (0,-2) -- (-2,0) --cycle;
\draw[thick,->] (2,0) --++ (-0.3535,0.3535);
\draw[thick,->] (3/2,1/2) --++ (-0.495,0.07);
\draw[thick,->] (1,1) --++ (-0.3535,-0.3535);
\draw[thick,->] (1/2,3/2) --++ (0.07,-0.495);
\draw[thick,->] (0,2) --++ (0.3535,-0.3535);
\draw[thick,->] (-1/2,3/2) --++ (0.495,0.07);
\draw[thick,->] (-1,1) --++ (0.3535,0.3535);
\draw[thick,->] (-3/2,1/2) --++ (-0.07,0.495);
\draw[thick,->] (-2,0) --++ (-0.3535,0.3535);
\draw[thick,->] (-3/2,-1/2) --++ (-0.495,0.07);
\draw[thick,->] (-1,-1) --++ (-0.3535,-0.3535);
\draw[thick,->] (-1/2,-3/2) --++ (0.07,-0.495);
\draw[thick,->] (0,-2) --++ (0.3535,-0.3535);
\draw[thick,->] (1/2,-3/2) --++ (0.495,0.07);
\draw[thick,->] (1,-1) --++ (0.3535,0.3535);
\draw[thick,->] (3/2,-1/2) --++ (-0.07,0.495);
\end{tikzpicture}
\caption{Boundary conditions after ``unfolding'' ...}\label{fig:HBinPPAD-2D-boundary}
\end{subfigure}\hfill
\begin{subfigure}[t]{0.5\textwidth}
\centering
\begin{tikzpicture}
\fill[white] (-2.5,-2.5) rectangle ++ (5,5);
\filldraw[fill=lightgray, draw=black] (0,2) -- (2,0) -- (0,-2) -- (-2,0) --cycle;
\draw (0,2) --++ (0,-4);
\draw (2,0) --++ (-4,0);
\draw (1,1) --++ (0,-2);
\draw (-1,1) --++ (0,-2);
\draw (-1,1) --++ (2,0);
\draw (-1,-1) --++ (2,0);
\draw (0,1) -- (-1,0);
\draw (1,0) -- (0,-1);
\draw (1,0) -- (0,1);
\draw (-1,0) -- (0,-1);

\draw[thick, ->] (7/4,3/4) -- (5/4,1/4);
\draw[thick, ->] (5/4,1/4) -- (3/4,3/4);
\draw[thick, ->] (3/4,3/4) -- (1/4,1/4);
\draw[thick, ->] (1/4,1/4) -- (1/4,-1/4);
\draw[thick, ->] (-7/4,-3/4) -- (-5/4,-1/4);
\draw[thick, ->] (-5/4,-1/4) -- (-5/4,1/4);

\node at (0,2.2) {2};
\node at (0,-2.2) {2};
\node at (2.2,0) {1};
\node at (-2.2,0) {1};
\node at (1.15,1.15) {2};
\node at (-1.15,-1.15) {2};
\node at (1.15,-1.15) {0};
\node at (-1.15,1.15) {0};
\node at (-0.15,1.2) {2};
\node at (-0.15,-1.2) {0};
\node at (-0.15,0.2) {2};
\node at (1.15,-0.2) {1};
\node at (-0.85,-0.4) {2};

\end{tikzpicture}
\caption{... yielding a Sperner instance with two sources}\label{fig:HBinPPAD-2D-triangulated}
\end{subfigure}
\caption{An example for the proof of~\cref{thm:HBinPPAD} in the case $k=2$. The region $C$ is represented in grey.}
\end{figure}

\begin{proof}[Proof of \cref{thm:HBinPPAD}]

Let $(f, \varepsilon)$ be an instance of \textsc{Hairy-Ball($\mathcal{F}$)} with $f: S^k \to \mathbb{R}^{k+1}$.

\medskip

\noindent \textbf{Stereographic projection.} Consider the unit sphere $S^k$ in $\mathbb{R}^{k+1}$. For convenience, we let the coordinate index start at $0$ in $\mathbb{R}^{k+1}$, i.e.\ $x = (x_0,x_1, \dots, x_k)$. Let $p = (1, 0, \dots, 0) \in S^k$. The stereographic projection with respect to the pole $p$ is defined as
$$SP : \quad S^k \setminus \{p\} \to \mathbb{R}^k, \quad (x_0, x_1, \dots, x_k) \mapsto \frac{1}{1-x_0} \left( x_1, \dots, x_k \right)$$
and its inverse is
$$SP^{-1} : \quad \mathbb{R}^k \to S^k \setminus \{p\}, \quad (z_1, \dots, z_k) \mapsto p +\frac{2}{1 + \sum_{i=1}^k z_i^2} (-1, z_1, \dots, z_k).$$
Note that the stereographic projection and its inverse can be computed exactly in polynomial time in the representation size of the rational input point.
In particular, rational points are always mapped to rational points.
Furthermore, it is easy to check that the inverse stereographic projection $SP^{-1}$ is $4\sqrt{k}$-Lipschitz continuous (w.r.t.\ $\ell_\infty$-norm).

\medskip

\noindent \textbf{Unfolding: changing the domain and range of $\boldsymbol{f}$.} We want to ``transform'' the function $f: S^k \to \mathbb{R}^{k+1}$ into a function $g: \mathbb{R}^k \to \mathbb{R}^k$, which is more convenient. Changing the domain of $f$ is easy: the stereographic projection ``unfolds'' $S^k \setminus \{p\}$ into $\mathbb{R}^k$. To change the range of $f$ we would like to also ``unfold'' the tangent vector field so that it now outputs a vector in $\mathbb{R}^k$ instead of $\mathbb{R}^{k+1}$. One way to achieve this is to find continuous vector fields $b_i : S^k \setminus \{p\} \to \mathbb{R}^{k+1}$, $i = 1, \dots, k$, such that for every $x \in S^k \setminus \{p\}$, $b_1(x), \dots, b_k(x)$ is a basis of the tangent space of $S^k$ at $x$. Expressing $f(x)$ in this local basis then yields an element in $\mathbb{R}^k$, as desired. We can explicitly construct such vector fields by using $SP$ and $SP^{-1}$ to map the standard basis of $\mathbb{R}^k$ into the tangent space at each $x \in S^k \setminus \{p\}$. We obtain
\begin{equation}\label{eq:basis-field}
\left[b_i(x)\right]_j = \left\{ \begin{tabular}{ll}
    $x_i$ & if $j = 0$ \\
    $1-x_i^2/(1-x_0)$ & if $j = i$ \\
    $-x_j x_i/(1-x_0)$ & otherwise
\end{tabular} \right.
\end{equation}
where $[\cdot]_j$ indicates the $j$th coordinate for $j \in \{0, 1, \dots, k\}$. It is straightforward to check that $b_1, \dots, b_k$ are continuous tangent vector fields that yield an orthonormal basis of the tangent space of $S^k$ at every $x \in S^k \setminus \{p\}$.

``Unfolding'' $f$ yields $g : \mathbb{R}^k \to \mathbb{R}^k$ which is defined as follows for $i=1,\dots, k$
$$[g(z)]_i = \langle f(x(z)), b_i(x(z)) \rangle$$
where we define $x(z) := SP^{-1}(z)$ for convenience. Intuitively, $g$ corresponds to the function that first maps $z \in \mathbb{R}^k$ to a point $x$ on the sphere using the inverse stereographic projection, computes $f(x)$ and then expresses $f(x)$ in the local basis $(b_1(x), \dots, b_k(x))$. Note that $g(z)$ can be computed in polynomial time in $\text{size}(f)$ and $\text{size}(z)$.

Since the $b_i$ always form an orthonormal basis, it follows that we always have $\|g(z)\|_2 = \|f(SP^{-1}(z))\|_2$. Thus, in order to find some $x \in S^k$ with $\|f(x)\|_\infty \leq \varepsilon$, it suffices to find some $z \in \mathbb{R}^k$ with $\|g(z)\|_\infty \leq \varepsilon/\sqrt{k}$.

\medskip

\noindent \textbf{Continuity of $\boldsymbol{g}$.} Clearly, $g$ is continuous. Moreover, since $\mathcal{F}$ is polynomially continuous, we can extend this to $g$ in the following sense.

\begin{claim}\label{clm:inPPAD:g-cont}
There exists a polynomial $r$ (that only depends on $\mathcal{F}$) such that for any $\widehat{\varepsilon} > 0$, there exists $\widehat{\delta}$ with $\text{size}(\widehat{\delta}) \leq r(\text{size}(f) + \text{size}(\widehat{\varepsilon}))$ such that for any $z,z' \in \mathbb{R}^k$ we have $\|z-z'\|_\infty \leq \widehat{\delta}$ $\implies$ $\|g(z)-g(z')\|_\infty \leq \widehat{\varepsilon}$.
\end{claim}

\begin{claimproof}
First, let us prove a bound on $\|f(x)\|_\infty$ for all $x \in S^k$. Let $\Bar{\delta} > 0$ be such that $\|f(x)-f(y)\|_\infty \leq 1$ if $\|x-y\|_\infty \leq \Bar{\delta}$. Recall that $\text{size}(\Bar{\delta})$ is polynomially bounded in $\text{size}(f)$. Since $f$ is continuous on the sphere $S^k$, there exists some $x^* \in S^k$ such that $f(x^*) = 0$. The arc distance between $x$ and $x^*$ is at most $\pi$ and we pick points $y^{(0)}, y^{(1)}, \dots, y^{(m)} \in S^k$ along the arc such that $y^{(0)}=x, y^{(m)}=x^*$ and $\|y^{(i)}-y^{(i+1)}\|_\infty \leq \|y^{(i)}-y^{(i+1)}\|_2 \leq \Bar{\delta}$. We can take $m \leq \lceil \pi / \Bar{\delta} \rceil \leq 4 / \Bar{\delta}$ (assuming $\Bar{\delta} \leq 1/2$). Then, we have
$$\|f(x)\|_\infty = \|f(x)-f(x^*)\|_\infty \leq \sum_{i=0}^{m-1} \|f(y^{(i)}) - f(y^{(i+1)})\|_\infty \leq \frac{4}{\Bar{\delta}}.$$
For any $z,z' \in \mathbb{R}^k$ we have for all $i$
\begin{equation*}\begin{split}
    |[g(z)]_i - [g(z')]_i| &= |\langle f(x(z)), b_i(x(z)) \rangle - \langle f(x(z')), b_i(x(z')) \rangle|\\
    &\leq |\langle f(x(z)) - f(x(z')), b_i(x(z)) \rangle| + |\langle f(x(z')), b_i(x(z)) - b_i(x(z')) \rangle|\\
    &\leq \|f(x(z)) - f(x(z'))\|_2 + \|f(x(z'))\|_2 \|b_i(x(z)) - b_i(x(z'))\|_2\\
\end{split}\end{equation*}
Note that for all $i$, $z \mapsto b_i(x(z))$ is Lipschitz-continuous with a Lipschitz constant of the form $m \sqrt{k}$ for some constant $m$ ($m = 20$ is enough). This can be checked by direct computation. Recalling that $x(z)$ is short for $SP^{-1}(z)$, the fact that $SP^{-1}$ is $4 \sqrt{k}$-Lipschitz and $f$ is polynomially continuous, the claim then follows.
\end{claimproof}

\medskip

\noindent \textbf{The colouring.} The function $g$ induces a colouring on $\mathbb{R}^k$. Every $z \in \mathbb{R}^k$ is assigned a colour as follows. First, compute $u := g(z) \in \mathbb{R}^k$. If $u_i \geq 0$ for all $i$, then assign colour $0$. Otherwise assign colour $j = \argmin_{i} u_i$ (break ties by picking the smallest such index).

Pick $\delta > 0$ so that $\|z-z'\|_\infty \leq \delta$ implies
\begin{equation}\label{eq:HBinPPAD-delta-cont}
\|g(z)-g(z')\|_\infty \leq \frac{\varepsilon}{8 \sqrt{k}}
\end{equation}
By \cref{clm:inPPAD:g-cont}, it is possible to pick such $\delta$ with $\text{size}(\delta) \leq \text{poly}(\text{size}(f),\text{size}(\varepsilon))$.

A panchromatic $\delta$-fine $k$-simplex in $\mathbb{R}^k$ is a $k$-simplex $z^{(0)}, \dots, z^{(k)}$ in $\mathbb{R}^k$, such that $z^{(i)}$ has colour $i$ and $\|z^{(i)}-z^{(j)}\|_\infty \leq \delta$ for all $i,j$.
\begin{claim}\label{clm:inPPAD:panchromatic}
Any panchromatic $k$-simplex in $\mathbb{R}^k$ yields a solution.
\end{claim}

\begin{claimproof}
Let the $\delta$-fine $k$-simplex $z^{(0)}, \dots, z^{(k)}$ be panchromatic. Since $[g(z^{(0)})]_j \geq 0$ for all $j$, there exists $i$ such that $[g(z^{(0)})]_i = \|g(z^{(0)})\|_\infty$. However, it must hold that $[g(z^{(i)})]_i < 0$ and thus $\|g(z^{(0)})\|_\infty \leq |[g(z^{(0)})]_i- [g(z^{(i)})]_i|$. Since $\|z^{(0)} - z^{(i)}\|_\infty \leq \delta$, it follows by \eqref{eq:HBinPPAD-delta-cont} that  $\|g(z^{(0)})\|_\infty \leq \varepsilon/8\sqrt{k} \leq \varepsilon/\sqrt{k}$, i.e.\ $z^{(0)}$ yields a solution.
\end{claimproof}

\medskip

\noindent \textbf{Sperner.} The next step is to show that we can locate a panchromatic $k$-simplex by using a Sperner-like argument. We explain the intuition in the case $k=2$. If $\|f(p)\|_\infty \leq \varepsilon$, then we have found a solution. If this is not the case, then we consider a sufficiently large region $C$ in $\mathbb{R}^k$ centred at $0$. We can show that on the boundary of $C$, $g$ (approximately) behaves as in \cref{fig:HBinPPAD-2D-boundary}, which yields a colouring similar to \cref{fig:HBinPPAD-2D-triangulated}. The colouring on the boundary of the domain in \cref{fig:HBinPPAD-2D-triangulated} has a very special structure. First of all, it is antipodally symmetric, namely, points lying on opposite sides have the same colour. Furthermore, if we start from any point on the boundary and ``move'' along the boundary in anti-clockwise direction, then we will see colour 1 change into colour 2 exactly twice. However, we will never see colour 2 change into colour 1; there will always be colour 0 in-between.

Intuitively, we want to pick a $\delta$-fine triangulation of the ball and construct a directed graph on the triangles. There is a directed edge between two triangles if they share a facet coloured $1,2$ and the edge is directed such that it crosses the facet with the colour $1$ on its left-hand side and the colour $2$ is on its right-hand side. Then, because of the structure of the colouring on the boundary, we obtain two known sources and any sink or other source of this directed graph corresponds to a panchromatic simplex, i.e.\ a solution.

This intuition is formalised in the following lemma, which is proved below.

\begin{lemma}\label{lem:HBinPPAD:reduce2sEOL}
The problem of finding a $\delta$-fine panchromatic $k$-simplex in $\mathbb{R}^k$ reduces to $2$-source \textsc{End-of-Line}.
\end{lemma}
By \cref{thm:mulsourceeol} in \cref{sec:EOL} it then follows that the problem lies in PPAD, which concludes the proof of \cref{thm:HBinPPAD}.
\end{proof}

\begin{proof}[Proof of \cref{lem:HBinPPAD:reduce2sEOL}]
If $\|f(p)\|_\infty \leq \varepsilon$, then we have found a solution. Thus, assume that this is not the case. To simplify the analysis we assume that the standard coordinate system is such that $f(p) = (0,v,v, \dots, v)/\sqrt{k}$, $v > \varepsilon$. If this is not the case then we perform a rotation such that this holds (at least approximately; a small error can be tolerated). Let $e^{(i)}$ be the $i$th unit vector in $\mathbb{R}^k$.

In order to investigate the colouring, we first prove some properties of an \emph{ideal} colouring that we define below and then show that the \emph{actual} colouring is ``close'' enough to the ideal colouring and also satisfies these properties.

\medskip

\noindent \textbf{The ideal colouring.} The ideal colouring is a function $w : \mathbb{R}^k \setminus \{0\} \to \mathbb{R}^k$. For $z \in S^{k-1}$ it is given by $w_i(z) = (1-2z_i \sum_{j=1}^k z_j)/\sqrt{k}$ for $1 \leq i \leq k$. For $z \in \mathbb{R}^k \setminus \{0\}$, we set $w(z) := w(z/\|z\|_2)$. The ideal colour of $z$ is the colour corresponding to $w(z)$ (using the same procedure as for the actual colouring), instead of $g(z)$.

\begin{claim}\label{clm:inPPAD:ideal-properties}
The following properties of $w$ are easy to check by direct computation.
\begin{enumerate}
    \item $w$ is antipodally symmetric, i.e.\ $w(-z)=w(z)$ for all $z \in \mathbb{R}^k \setminus \{0\}$.
    \item For all $z \in \mathbb{R}^k \setminus \{0\}$, $\|w(z)\|_2 = 1$.
    \item If $w(z) \leq 0$, then $z \geq 0$ or $z \leq 0$.
    \item For any subset of the indices $\emptyset \neq I \subseteq \{1, \dots, k\}$, if $z \in S^{k-1}$ is such that $z_i \geq 0$ for $i \in I$ and $z_i = 0$ otherwise, then $w_j(z)=0$ for all $j \notin I$ and there exists $i \in I$ such that $w_i(z) \leq -1/\sqrt{k}$. In particular, any colour-direction $u \in \mathbb{R}^k$ with $\|u-w(z)\|_\infty < 1/2\sqrt{k}$ yields a colour in $I$.
\end{enumerate}
\end{claim}

\medskip

\noindent \textbf{The ideal colouring is ``close'' to the actual colouring.} We now show that if $\|z\|_2$ is sufficiently large, then the ideal colouring is very close to the actual colouring. Namely, we will show that $w(z)$ is almost equal to $\widehat{g}(z)$, a ``normalisation'' of $g$ defined by $\widehat{g}(z) = c(z) \cdot g(z)$, where $c(z) = \frac{1+\|z\|_2^2}{v \|z\|_2^2}$. Note that $\widehat{g}$ and $g$ yield the exact same colour for any $z \neq 0$.

\begin{claim}\label{clm:inPPAD:ideal-actual-col}
There exists $t > 0$ with $\textup{size}(t) = \textup{poly}(\textup{size}(f), \textup{size}(\varepsilon))$ such that for all $z \in \mathbb{R}^k$ with $\|z\|_2 \geq t$
$$\|\widehat{g}(z)-w(z)\|_\infty \leq \frac{1}{4\sqrt{k}}.$$
\end{claim}

\begin{claimproof}
Since $f$ is polynomially continuous, there exists $\widehat{\delta}$ such that $\|f(x)-f(p)\|_\infty \leq \varepsilon/16\sqrt{k(k+1)}$ with $\textup{size}(\widehat{\delta}) = \textup{poly}(\textup{size}(f), \textup{size}(\varepsilon))$. Using the definition of $SP^{-1}$ it is easy to check that $\|x(z)-p\|_2 \leq 2/\sqrt{1+\|z\|_2^2}$. Thus, we can construct $t \geq 4$ with $\textup{size}(t) = \textup{poly}(\textup{size}(f), \textup{size}(\varepsilon))$ such that $\|f(x(z))-f(p)\|_\infty \leq \varepsilon/16\sqrt{k(k+1)}$ for all $z$ with $\|z\|_2 \geq t$.

Let $\Bar{g}(z)$ denote the colour-direction obtained at $z$ if we use $f(p)$, instead of $f(x(z))$, i.e.\ $\Bar{g}_i(z) = \langle f(p), b_i(x(z)) \rangle$. Then, for all $z$ with $\|z\|_2 \geq t$, we have
\begin{equation}\label{eq:HBinPPAD-using-pole-instead}
\|g(z) - \Bar{g}(z)\|_\infty \leq \sqrt{k+1} \|f(x(z))-f(p)\|_\infty \leq \varepsilon/16\sqrt{k}.
\end{equation}
Recalling that $f(p) = (0, v, v, \dots, v)/\sqrt{k}$ and using equation \eqref{eq:basis-field} and the definition of $SP^{-1}$, we can write for all $i$
\begin{align*}
\Bar{g}_i(z) = \langle f(p), b_i(x(z)) \rangle &= \frac{v}{\sqrt{k}} \sum_{j=1}^k [b_i(x(z))]_j\\
&= \frac{v}{\sqrt{k}} \left( 1 - \frac{x_i(z)}{1-x_0(z)} \sum_{j=1}^k x_j(z) \right)\\
&= \frac{v}{\sqrt{k}} \left( 1 - \frac{\frac{2z_i}{1+\|z\|_2^2}}{\frac{2}{1+\|z\|_2^2}} \sum_{j=1}^k \frac{2z_j}{1+\|z\|_2^2} \right)\\
&= \frac{v}{\sqrt{k}} \left( 1 - \frac{2z_i}{1+\|z\|_2^2} \sum_{j=1}^k z_j \right)\\
&= \frac{v}{\sqrt{k}} \cdot \frac{\|z\|_2^2}{1+\|z\|_2^2} \left( \frac{1+\|z\|_2^2}{\|z\|_2^2} - \frac{2z_i}{\|z\|_2^2} \sum_{j=1}^k z_j \right)\\
&= \frac{v \|z\|_2^2}{1+\|z\|_2^2} \left( \frac{1 - 2\frac{z_i}{\|z\|_2} \sum_{j=1}^k \frac{z_j}{\|z\|_2}}{\sqrt{k}} + \frac{1}{\sqrt{k} \|z\|_2^2} \right)\\
&= \frac{1}{c(z)}\left( w_i(z) + \frac{1}{\sqrt{k} \|z\|_2^2} \right).
\end{align*}
We obtain that for any $z$ with $\|z\|_2 \geq t$
\begin{equation*}\begin{split}
\|\widehat{g}(z) - w(z)\|_\infty \leq c(z) \|g(z)-\Bar{g}(z)\|_\infty + \|c(z) \cdot \Bar{g}(z) - w(z)\|_\infty &\leq \frac{c(z) \varepsilon}{16\sqrt{k}} + \frac{1}{\sqrt{k}\|z\|_2^2}\\
&\leq \frac{1+\|z\|_2^2}{16 \sqrt{k} \|z\|_2^2} + \frac{1}{\sqrt{k}t^2}\\
&\leq 1/4\sqrt{k}
\end{split}\end{equation*}
where we used \eqref{eq:HBinPPAD-using-pole-instead}, $v > \varepsilon$ and $t \geq 4$.
\end{claimproof}

\medskip

\noindent \textbf{The cross-polytope $\boldsymbol{C_m}$.} Instead of working on a ball in $\mathbb{R}^k$, we define our Sperner instance on a different region that is easier to triangulate efficiently while still behaving nicely with respect to the colouring. Let $m = \lceil t \sqrt{k} \rceil$. The \emph{cross-polytope} of radius $m$ is defined as $C_m = \{z \in \mathbb{R}^k : \|z\|_1 \leq m\}$.

The unit cross-polytope can be triangulated efficiently by using a standard efficient triangulation of the $k$-simplex $0,e^{(1)}, \dots, e^{(k)}$ (see e.g.~\cite{todd1976fixedpoints}) and then extending it to the rest of the cross-polytope by mirroring along each axis. Thus, using a $\delta/m$-fine efficient triangulation of the simplex, extending it to the cross-polytope and then scaling to $C_m$, yields an efficient $\delta$-fine triangulation of $C_m$. We call this triangulation $T$. In particular, given a simplex of $T$ and one of its facets, we can compute the other simplex of $T$ that shares this facet in polynomial time (or decide that it does not exist, if the facet lies on the boundary of $C_m$). Note that $\text{size}(\delta/m)$ is polynomial in the size of the input.

\medskip

\noindent \textbf{The boundary of $\boldsymbol{C_m}$.} The boundary of $C_m$ is denoted $\partial C_m = \{z \in \mathbb{R}^k : \|z\|_1 = m\}$. Note that for any $z \in \partial C_m$ we have $\|z\|_2 \geq t$. Let $\partial C_m^+ := \{z \geq 0 : \|z\|_1 = m\}$ and $\partial C_m^- := \{z \leq 0 : \|z\|_1 = m\}$ denote the intersection of $\partial C_m$ with the all-positive and all-negative orthant respectively. Note that $\partial C_m^+$ and $\partial C_m^-$ are the $(k-1)$-simplices $m \cdot e^{(1)}, \dots, m \cdot e^{(k)}$ and $-m \cdot e^{(1)}, \dots, -m \cdot e^{(k)}$. Furthermore, by construction, $T$ induces a triangulation on $\partial C_m^+$ and $\partial C_m^-$. We now prove two key properties of the actual colouring on $\partial C_m$.

\begin{claim}\label{clm:inPPAD:sperner-boundary}
$\partial C_m^+$ satisfies the standard Sperner boundary conditions. Namely, any face $m \cdot e^{(i_1)}, \dots, m \cdot e^{(i_\ell)}$ ($1 \leq i_1 < \dots < i_\ell \leq k$, $1 \leq \ell \leq k$) of $\partial C_m^+$ only contains colours in $\{i_1, \dots, i_\ell\}$ in the actual colouring. The same also holds for $\partial C_m^-$.
\end{claim}

\begin{claimproof}
By \cref{clm:inPPAD:ideal-properties} (Item 4), it suffices to show that $\|\widehat{g}(z)-w(z)\|_\infty < 1/2\sqrt{k}$ for all $z$ on the face $m \cdot e^{(i_1)}, \dots, m \cdot e^{(i_\ell)}$, which holds by \cref{clm:inPPAD:ideal-actual-col} and the choice of $m$.
\end{claimproof}

\begin{claim}\label{clm:inPPAD:no-other-entrance}
The restriction of the triangulation $T$ to $\partial C_m \setminus (\partial C_m^+ \cup \partial C_m^-)$ does not contain any $(k-1)$-simplex coloured $1,2,\dots, k$ (in the actual colouring).
\end{claim}

\begin{claimproof}
Let $y^{(1)}, \dots, y^{(k)}$ be a $\delta$-fine $(k-1)$-simplex on $\partial C_m$ such that $y^{(i)}$ has colour $i$. By \cref{clm:inPPAD:ideal-properties} (Item 3) it suffices to show that $w(y^{(i)}) \leq 0$ for all $i$. Note that since $\|y^{(i)} - y^{(j)}\|_\infty \leq \delta$, it follows by \eqref{eq:HBinPPAD-delta-cont}
\begin{equation}\label{eq:HBinPPAD-normalised-g-cont}
    \|\widehat{g}(y^{(i)}) - \widehat{g}(y^{(j)})\|_\infty \leq \frac{1}{8\sqrt{k}} \frac{\varepsilon}{v} \frac{1+\|z\|_2^2}{\|z\|_2^2} \leq 1/4\sqrt{k}
\end{equation}
since $v > \varepsilon$ and $\|z\|_2 \geq t \geq 1$.

Now assume that there exists $\ell$ such that $w_\ell(y^{(i)}) > 0$. \eqref{eq:HBinPPAD-normalised-g-cont} and \cref{clm:inPPAD:ideal-actual-col} imply that $\widehat{g}_\ell(y^{(\ell)}) > -1/2\sqrt{k}$ and thus also $\widehat{g}_j(y^{(\ell)}) > -1/2\sqrt{k}$ for all $j$, since $y^{(\ell)}$ has colour $\ell$ in the actual colouring. On the other hand, since $\widehat{g}_j(y^{(j)}) < 0$ it follows that $\widehat{g}_j(y^{(\ell)}) < 1/4\sqrt{k}$ using \eqref{eq:HBinPPAD-normalised-g-cont}. Thus, we get $\|\widehat{g}(y^{(\ell)})\|_\infty < 1/2\sqrt{k}$. Using \cref{clm:inPPAD:ideal-actual-col} again, it follows that $\|w(y^{(\ell)})\|_\infty < 1/\sqrt{k}$, which implies $\|w(y^{(\ell)})\|_2 < 1$, a contradiction to \cref{clm:inPPAD:ideal-properties}.
\end{claimproof}

\medskip

\noindent \textbf{Orientation of $\boldsymbol{\partial C_m^+}$ and $\boldsymbol{\partial C_m^-}$.} The following simple observation is crucial.

\begin{claim}\label{clm:inPPAD:orientation}
The $(k-1)$-simplices $\partial C_m^+$ and $\partial C_m^-$ have the same orientation with respect to the Sperner boundary conditions.
\end{claim}

\begin{claimproof}
Consider the linear function $T : \mathbb{R}^k \to \mathbb{R}^k$, $T(y) = -y$ and note that $T$ maps $\partial C_m^+$ to $\partial C_m^-$ while preserving the Sperner boundary conditions (\cref{clm:inPPAD:sperner-boundary}). Furthermore, the determinant of $T$ is $(-1)^k = 1$, since $k$ is even. It follows that $\partial C_m^+$ and $\partial C_m^-$ have the same orientation with respect to the Sperner boundary conditions.
\end{claimproof}

This is the only point in the proof where we use the fact that $k$ is even. However, it is indeed crucial, since this ensures that the two ends of line that we know are both sources (or both sinks), as we will see below. If $k$ is odd, our reduction yields an instance with a known source and sink. In this case, there is no guarantee that a solution -- another end of line -- even exists.

\begin{figure}[h!]
\centering
\begin{tikzpicture}
\filldraw[fill=lightgray, draw=black, thick] (0,5) -- (5,0) -- (0,-5) -- (-5,0) --cycle;

\filldraw[gray] (5,5) circle (2pt) node[anchor=south west] {2};
\draw[gray, thick] (5,5) -- (0,5);
\draw[gray, thick] (5,5) -- (0.5,4.5);
\draw[gray, thick] (5,5) -- (1,4);
\draw[gray, thick] (5,5) -- (1.5,3.5);
\draw[gray, thick] (5,5) -- (2,3);
\draw[gray, thick] (5,5) -- (2.5,2.5);
\draw[gray, thick] (5,5) -- (3,2);
\draw[gray, thick] (5,5) -- (3.5,1.5);
\draw[gray, thick] (5,5) -- (4,1);
\draw[gray, thick] (5,5) -- (4.5,0.5);
\draw[gray, thick] (5,5) -- (5,0);

\filldraw[gray] (-5,-5) circle (2pt) node[anchor=north east] {2};
\draw[gray, thick] (-5,-5) -- (0,-5);
\draw[gray, thick] (-5,-5) -- (-0.5,-4.5);
\draw[gray, thick] (-5,-5) -- (-1,-4);
\draw[gray, thick] (-5,-5) -- (-1.5,-3.5);
\draw[gray, thick] (-5,-5) -- (-2,-3);
\draw[gray, thick] (-5,-5) -- (-2.5,-2.5);
\draw[gray, thick] (-5,-5) -- (-3,-2);
\draw[gray, thick] (-5,-5) -- (-3.5,-1.5);
\draw[gray, thick] (-5,-5) -- (-4,-1);
\draw[gray, thick] (-5,-5) -- (-4.5,-0.5);
\draw[gray, thick] (-5,-5) -- (-5,0);

\filldraw[black] (5,0) circle (2pt) node[anchor=north west] {1};
\filldraw[black] (4.5,0.5) circle (2pt) node[anchor=north east] {1};
\filldraw[black] (4,1) circle (2pt) node[anchor=north east] {1};
\filldraw[black] (3.5,1.5) circle (2pt) node[anchor=north east] {1};
\filldraw[black] (3,2) circle (2pt) node[anchor=north east] {2};
\filldraw[black] (2.5,2.5) circle (2pt) node[anchor=north east] {2};
\filldraw[black] (2,3) circle (2pt) node[anchor=north east] {1};
\filldraw[black] (1.5,3.5) circle (2pt) node[anchor=north east] {2};
\filldraw[black] (1,4) circle (2pt) node[anchor=north east] {2};
\filldraw[black] (0.5,4.5) circle (2pt) node[anchor=north east] {2};
\filldraw[black] (0,5) circle (2pt) node[anchor=south east] {2};

\filldraw[black] (-0.5,4.5) circle (2pt) node[anchor=south east] {2};
\filldraw[black] (-1,4) circle (2pt) node[anchor=south east] {0};
\filldraw[black] (-1.5,3.5) circle (2pt) node[anchor=south east] {2};
\filldraw[black] (-2,3) circle (2pt) node[anchor=south east] {0};
\filldraw[black] (-2.5,2.5) circle (2pt) node[anchor=south east] {0};
\filldraw[black] (-3,2) circle (2pt) node[anchor=south east] {0};
\filldraw[black] (-3.5,1.5) circle (2pt) node[anchor=south east] {0};
\filldraw[black] (-4,1) circle (2pt) node[anchor=south east] {0};
\filldraw[black] (-4.5,0.5) circle (2pt) node[anchor=south east] {1};

\filldraw[black] (-5,0) circle (2pt) node[anchor=south east] {1};
\filldraw[black] (-4.5,-0.5) circle (2pt) node[anchor=south west] {1};
\filldraw[black] (-4,-1) circle (2pt) node[anchor=south west] {1};
\filldraw[black] (-3.5,-1.5) circle (2pt) node[anchor=south west] {1};
\filldraw[black] (-3,-2) circle (2pt) node[anchor=south west] {2};
\filldraw[black] (-2.5,-2.5) circle (2pt) node[anchor=south west] {1};
\filldraw[black] (-2,-3) circle (2pt) node[anchor=south west] {1};
\filldraw[black] (-1.5,-3.5) circle (2pt) node[anchor=south west] {2};
\filldraw[black] (-1,-4) circle (2pt) node[anchor=south west] {2};
\filldraw[black] (-0.5,-4.5) circle (2pt) node[anchor=south west] {2};
\filldraw[black] (0,-5) circle (2pt) node[anchor=north west] {2};

\filldraw[black] (0.5,-4.5) circle (2pt) node[anchor=north west] {2};
\filldraw[black] (1,-4) circle (2pt) node[anchor=north west] {2};
\filldraw[black] (1.5,-3.5) circle (2pt) node[anchor=north west] {0};
\filldraw[black] (2,-3) circle (2pt) node[anchor=north west] {0};
\filldraw[black] (2.5,-2.5) circle (2pt) node[anchor=north west] {0};
\filldraw[black] (3,-2) circle (2pt) node[anchor=north west] {0};
\filldraw[black] (3.5,-1.5) circle (2pt) node[anchor=north west] {1};
\filldraw[black] (4,-1) circle (2pt) node[anchor=north west] {1};
\filldraw[black] (4.5,-0.5) circle (2pt) node[anchor=north west] {0};

\draw[blue,thick, ->] (5.5,1) -- (4.75,1);
\draw[blue,thick, ->] (4.7,1.05) -- (4.35,1.4);
\draw[blue,thick, ->] (4.3,1.45) -- (3.95,1.8);
\draw[blue,thick, ->] (3.9,1.85) -- (3.55,2.2);
\draw[blue,thick, ->] (3.5,2.15) -- (2.8,1.2);
\draw[blue,thick, dotted] (2.8,1.2) -- (2.3,0.7);

\draw[blue,thick, dotted] (1.3,1.8) -- (1.8,2.3);
\draw[blue,thick, ->] (1.8,2.3) -- (2.55,3.05);
\draw[blue,thick, ->] (2.6,3.1) -- (2.2,3.5);
\draw[blue,thick, ->] (2.15,3.5) -- (1.4,2.9);
\draw[blue,thick, dotted] (1.4,2.9) -- (0.9,2.4);

\draw[blue,thick, ->] (-5.5,-1) -- (-4.75,-1);
\draw[blue,thick, ->] (-4.7,-1.05) -- (-4.35,-1.4);
\draw[blue,thick, ->] (-4.3,-1.45) -- (-3.95,-1.8);
\draw[blue,thick, ->] (-3.9,-1.85) -- (-3.55,-2.2);
\draw[blue,thick, ->] (-3.5,-2.15) -- (-2.8,-1.3);
\draw[blue,thick, dotted] (-2.8,-1.3) -- (-2.3,-0.8);

\draw[blue,thick, dotted] (-1.8,-1.3) -- (-2.3,-1.8);
\draw[blue,thick, ->] (-2.3,-1.8) -- (-3,-2.6);
\draw[blue,thick, ->] (-3.05,-2.65) -- (-2.65,-3.05);
\draw[blue,thick, ->] (-2.6,-3.1) -- (-2.2,-3.5);
\draw[blue,thick, ->] (-2.15,-3.5) -- (-1.4,-2.9);
\draw[blue,thick, dotted] (-1.4,-2.9) -- (-0.9,-2.4);

\filldraw[red] (-5.5,-1) circle (3pt);
\filldraw[red] (5.5,1) circle (3pt);

\end{tikzpicture}
\caption{Example in 2 dimensions showing how the ``artificial start'' trick is used to obtain a 2-source \textsc{End-of-Line} instance. The cross-polytope $C_m$ is shown in grey and the vertices of the triangulation of $C_m$ that lie on $\partial C_m$ are also shown together with their colours. The ``artificial start'' trick has introduced a new vertex with colour 2 that is connected to all the triangulation vertices on $\partial C_m^+$. The same has also been done for $\partial C_m^-$. The two new vertices and the new edges are shown in light grey colour. All the directed edges of the resulting \textsc{End-of-Line} instance that cross the boundary of $C_m$ or lie outside $C_m$ are shown in blue in this example. The two red points represent the two sources we obtain for the \textsc{End-of-Line} instance. Note that in this example the colouring is not perfectly antipodally symmetric, since this is not the case in general. However, it is easy to see that the displayed colouring satisfies \cref{clm:inPPAD:sperner-boundary}, \cref{clm:inPPAD:no-other-entrance} and \cref{clm:inPPAD:orientation}.
}\label{fig:sperner-boundary-trick}
\end{figure}

\medskip

\noindent \textbf{2-source \textsc{End-of-Line}.} We now explain how the problem of finding a panchromatic $k$-simplex in the triangulation $T$ of $C_m$ reduces to a 2-source \textsc{End-of-Line} problem. Let us briefly recall how a standard $k$-dimensional \textsc{Sperner} instance reduces to \textsc{End-of-Line}. A fully detailed and formal presentation of this is given in~\cite[Proposition 2.2]{etessami2010fixp}. In the standard $k$-dimensional \textsc{Sperner} problem, we have a big $k$-simplex triangulated into $k$-simplices and a colouring that satisfies the Sperner boundary conditions. The \textsc{End-of-Line} graph is constructed as follows. Every $k$-simplex of the triangulation is a node in the graph. There is an edge between two nodes if the corresponding $k$-simplices have a common facet coloured $1, \dots, k$ (called a \emph{door-facet}). The edge is directed by considering the orientation of the $k$-simplices with respect to the common facet. This yields an \textsc{End-of-Line} graph where every degree-1 node is either a panchromatic $k$-simplex or a $k$-simplex that has a door-facet lying on the boundary of the instance, i.e.\ on the boundary of the big $k$-simplex. Note that all the door-facets lying on the boundary of the instance lie on the $1, \dots, k$-facet of the big $k$-simplex. In order to obtain a single source, there are various standard tricks, see e.g.~\cite[Chapter 4]{todd1976fixedpoints}. If we use the so-called ``artificial start'' trick, then as shown in~\cite{todd1976fixedpoints,etessami2010fixp}, we obtain a slightly larger \textsc{Sperner} instance that has a single door-facet ``exposed'' to the outside, i.e.\ lying on the boundary of the instance.

We can use these standard techniques on our problem. The boundary of our instance is $\partial C_m$. By \cref{clm:inPPAD:no-other-entrance}, all door-facets on the boundary are contained in $\partial C_m^+$ or $\partial C_m^-$. By applying the trick mentioned above on $\partial C_m^+$ (interpreted as a facet of a $k$-dimensional \textsc{Sperner} instance) we obtain a single door-facet that is ``exposed'' in this whole region. Applying the exact same trick on $\partial C_m^-$ also yields a single door-facet that is ``exposed'' in this whole region. Thus, overall we only have two door-facets that lie on the boundary of our instance. \cref{fig:sperner-boundary-trick} shows how the ``artificial start'' trick is used on a 2-dimensional example.

Finally, the important observation here is that the orientation of the ``exposed'' door-facet obtained by using the trick only depends on the orientation of the big facet on which it is used. Since $\partial C_m^+$ and $\partial C_m^-$ have the same orientation (\cref{clm:inPPAD:orientation}), the two ``exposed'' door-facets have the same orientation. This means that we have two ``entrances'' for the standard Sperner path-following algorithm. By applying the formal arguments in the reduction from \textsc{Sperner} to \textsc{End-of-Line} this yields a two-source \textsc{End-of-Line} instance.
\end{proof}

\section{Computational Hardness for Hairy Ball Problems}\label{sec:HBhardness}

It is possible to prove Brouwer's fixed point theorem using the Hairy Ball Theorem as follows (see~\cite{milnor1978hairyball} for the full details). Let $B_k \subset \mathbb{R}^k$ be the unit ball. If we assume that a function $f: B_k \to B_k$ does not have any fixed point, then we can construct a Hairy Ball function $g: S^k \to \mathbb{R}^{k+1}$ that does not have a zero. Brouwer's theorem follows by contradiction. The main idea for the construction of $g$ is the following. Consider $f'(x)=f(x)-x$ and assume that it points directly inward on the boundary of $B_k$. Take one copy of $B_k$ with the vector field $f'$ and one copy with the vector field $-f'$, and glue their boundaries together. The resulting object can be deformed to yield the sphere $S^k$ and the vector fields will perfectly match on the glued region. Thus, assuming that $f'$ has no zero, $g$ will have no zero either.

By making this idea fully constructive and efficient, we obtain reductions from Brouwer problems to Hairy Ball problems. Thus, existing PPAD- and FIXP-hardness results for Brouwer also hold for the corresponding Hairy Ball problems. We note that these reductions always involve using \emph{two} copies of a Brouwer instance to obtain a single Hairy Ball instance. This further supports our claim that the fact that we obtain two sources when reducing \textsc{$k$D-Hairy-Ball} to \textsc{End-of-Line} (\cref{sec:HBinPPAD}) is not an artefact of our reduction, but intrinsic to the problem.

\subsection{PPAD-hardness}

\begin{theorem}\label{thm:kHB-PPAD-hard}
For all even $k \geq 2$, \textsc{$k$D-Hairy-Ball} is \textup{PPAD}-hard.
\end{theorem}

Note that this result is particularly strong, because the type of circuit allowed in the definition of \textsc{$k$D-Hairy-Ball} (\cref{def:linearHB}) is particularly weak. Furthermore, the hardness is proved for inversely exponential $\varepsilon/L$ (where $L$ is the Lipschitz constant of the function), which is the best we can hope for in the fixed dimension case. Indeed, if $\varepsilon/L$ is inversely polynomial, then the following is a polynomial time algorithm that solves the problem: divide the domain into a sufficiently small (but polynomial) number of regions and check an arbitrary point in each region.

Our reduction essentially shows that some hardness results known for the Brouwer problem also hold for the Hairy Ball problem. One direction that we do not explore here is the hardness of approximation when the dimension is not fixed. It seems quite likely that the problem should be hard even for constant $\varepsilon/L$ in that case, since this is known to hold for the Brouwer problem~\cite{rubinstein2018inapproximability}.

\begin{proof}[Proof Overview]
One way to prove this result is to take an instance $F$ of a \textsc{2D-Brouwer} problem, which is known to be PPAD-complete~\cite{chen2009complexity}, and embed a copy $F$ and a copy $-F$ on the south and north hemisphere of $S^2$ respectively. However, since our \textsc{$2$D-Hairy-Ball} circuit can only use gates in $\{+,\times \zeta, \max\}$, we first shrink the domain of $F$ so that we embed it in a small region around the south pole. This ensures that even after projection onto the tangent space, no bogus solutions will appear. We do the same for $-F$ in the north pole and then define $G$ on the rest of the sphere in such a way that no solution appears there.

In spirit, we follow this general proof idea, but take a slight detour, because it gives us the chance to define and study a discrete analog to \textsc{$2$D-Hairy-Ball}: the \textsc{$2$D-Hairy-Cube} problem. Intuitively, this problem is obtained from \textsc{$2$D-Hairy-Ball} the same way that discrete \textsc{2D-Brouwer} is obtained from continuous \textsc{2D-Brouwer}. The domain is discretised by a grid and a circuit computes the local direction of the function in every grid-square. The natural way to discretise the sphere $S^2$ is to replace it by a cube with a grid on each face. The advantage of \textsc{$2$D-Hairy-Cube} is that PPAD-hardness is easy to prove: just put a (slightly modified) discrete \textsc{2D-Brouwer} instance on one face, and the inverse instance on the opposite face. Defining the instance on the remaining faces is trivial in this case. \cref{sec:hairycube} defines the problem and proves PPAD-hardness.

The next step is reducing \textsc{$2$D-Hairy-Cube} to \textsc{$2$D-Hairy-Ball}. Even though it seems natural that this should hold, the reduction is technically tedious. In particular, we have to simulate a Boolean circuit using an arithmetic circuit, but the input bits cannot always be computed exactly. Thus, we need to use an averaging trick (see~\cite{daskalakis2009complexity,CDT}). The details are in \cref{sec:app:cubetoHB}.

This yields PPAD-hardness for \textsc{$2$D-Hairy-Ball}. The final step is to extend this to \textsc{$k$D-Hairy-Ball}, by showing that \textsc{$k$D-Hairy-Ball} reduces to \textsc{$(k+2)$D-Hairy-Ball}. The proof can be found in \cref{sec:app:kHBtok+2HB}.
\end{proof}

\subsection{FIXP-hardness}\label{sec:HBFIXPhard}

Up to this point we have only considered the problem of computing an \emph{approximate} Hairy Ball solution. However, there are other computational problems one could consider, e.g.\ computing a point that is close to an \emph{exact} solution, or computing the first $n$ bits of an \emph{exact} solution.

The corresponding problems for Brouwer fixed points have been studied by Etessami and Yannakakis~\cite{etessami2010fixp}. In particular, they define the class FIXP that captures the complexity of computing an exact fixed point of a function given by an arithmetic circuit and mapping the unit cube into itself. They prove that computing an exact Nash equilibrium of a 3-player game is FIXP-complete. In doing so, they use a special type of reduction, called an SL-reduction, that ensures that the reduction also holds for three problems that can be studied in the standard Turing machine model: the ``strong approximation problem'' (i.e.\ find a point close to an exact solution), the ``partial computation'' problem (i.e.\ compute the first $n$ bits of an exact solution) and various decision problems. This means that computing a strong approximation for 3-player Nash is as hard as computing a strong approximation of a Brouwer function given by an arithmetic circuit. We prove that the corresponding problems for the Hairy Ball Theorem are at least as hard as their Brouwer counterparts. For the formal definition of the class FIXP, see~\cite[p. 2570, Lemma 4.1, Theorem 4.7]{etessami2010fixp}, and for the formal definition of SL-reductions, see~\cite[p. 2540]{etessami2010fixp}. We define our exact computation problem for the Hairy Ball Theorem as follows:

\begin{definition}
The \textsc{Exact-Hairy-Ball} problem is defined as: given an arithmetic circuit $F$ (with gates $\{+,-,\times, /, \max,\min\}$, rational constants and that never divides by $0$) that computes a tangent vector field $S^k \to \mathbb{R}^{k+1}$, $k$ even, find $x \in S^k$ such that $F(x)=0$.
\end{definition}
Note that the vector field will be continuous since we never divide by $0$. The vector field can be syntactically forced to be tangent to the sphere because we can compute the projection exactly using this kind of circuit.

\begin{theorem}\label{thm:HBFIXPhard}
\textsc{Exact-Hairy-Ball} is \textup{FIXP}-hard. Furthermore, the corresponding strong approximation, partial computation and decision problems are also hard for the corresponding versions of \textup{FIXP} (as defined in~\cite{etessami2010fixp}).
\end{theorem}

Following Etessami and Yannakakis, we could define a class HB that captures the complexity of computing an exact Hairy Ball solution (and corresponding versions of the class for the related problems) by taking the set of all problems that reduce to \textsc{Exact-Hairy-Ball}. Then, \cref{thm:HBFIXPhard} is saying that FIXP $\subseteq$ HB. Note that the three discrete problems that can be studied in the Turing machine model lie in PSPACE, by using the same technique as in \cite[Proposition 4.2]{etessami2010fixp}.

\begin{proof}[Proof of \cref{thm:HBFIXPhard}]
The proof is based on an idea used by Milnor~\cite{milnor1978hairyball} to show that the Hairy Ball theorem implies Brouwer's fixed point theorem. Note, however, that Milnor's proof proceeds by contradiction, whereas our proof is fully constructive.

We embed two copies of a Brouwer instance (after some preprocessing) on the sphere such that any exact Hairy Ball solution yields an exact Brouwer fixed point. Note that this reduction makes use of the $\times$ and $/$ gates and cannot be used to prove PPAD-hardness of \textsc{$k$D-Hairy-Ball}.

Let $G$ be an arithmetic circuit with gates $\{+,-,\times, /, \max,\min\}$ (and rational constants) that maps the unit cube $C_k=\{x \in \mathbb{R}^k : \|x\|_\infty \leq 1\}$ into itself and does not divide by $0$. Computing a fixed point of $G$ is a FIXP-complete problem and we will reduce this to \textsc{Exact-Hairy-Ball}. At the end, we also explain why our reduction is indeed an SL-reduction.

\medskip

\noindent \textbf{Step 1: Preprocessing.} First of all, note that we can assume that $k$ is even. Indeed, we can always consider $\Bar{G} : C_{k+1} \to C_{k+1}$, $\Bar{G}(x_0,x) = (0,G(x))$. Note that fixed points are mapped one-to-one by dropping the first coordinate. Let $\ell$ be large enough such that $2^\ell \geq 2k$. For any integer $m$ let $mB_k = \{x \in \mathbb{R}^k : \|x\|_2 \leq m\}$. We construct the circuit $G'$ that on input $x \in 2^\ell B_k$ does:
\begin{enumerate}
    \item Compute $y=\frac{x}{\max\{1,\|x\|_\infty\}}$
    \item Output $-x \cdot \min\{1,\|x-y\|_\infty\} + (G(y)-x) \cdot \max\{0,1-\|x-y\|_\infty\}$
\end{enumerate}
It is easy to check that the fixed points of $G$ are exactly the zeros of $G'$. Furthermore, if $\|x\|_2 \geq 2^{\ell-1}$, then $G'(x)=-x$. Next, we change the domain from $2^\ell B_k$ to $1B_k=:B_k$. The circuit $\widehat{G}$ is constructed such that $\widehat{G}(x) = 2^{-\ell} G(2^{\ell} x)$. The zeros of $G'$ are exactly the zeros of $\widehat{G}$ scaled by $2^{\ell}$. Note that if $\|x\|_2 \geq 1/2$, we have $\widehat{G}(x)=-x$.

\medskip

\noindent \textbf{Step 2: From Brouwer to Hairy Ball.} The high-level idea in order to construct a Hairy Ball vector field is to use the output of $\widehat{G}(x)$ as coordinates in some continuous tangent basis on $S^k$. However, as noted before, it is impossible to define a continuous tangent basis on all of $S^k$. As a result, we use two different bases (namely $\widehat{b}_i(x)$ and $-\widehat{b}_i(-x)$) and obtain two tangent vectors at $x$: $w^-$ and $w^+$. The final output $u$ is then obtained by taking a convex combination of $w^-$ and $w^+$, according to a parameter $\alpha$. This parameter $\alpha$ ensures that we do not use a tangent basis when we are close to the pole where it is not a basis. Furthermore, it also ensures that the case where $u$ is obtained as a \emph{strict} convex combination (namely $\alpha \in (0,1)$) can only occur at points where $\widehat{G}(x)=-x$, which makes this case easier to analyse.

Formally, we construct the circuit $F:S^k \to \mathbb{R}^{k+1}$ that does the following on input $x=(x_0,x_1,\dots,x_k)$
\begin{enumerate}
    \item Compute $v = \widehat{G}(x_1, \dots, x_k)$
    \item Compute $w^- = \sum_{i=1}^k v_i \widehat{b}_i(x)$ and $w^+ = \sum_{i=1}^k (-v_i) \widehat{b}_i(-x)$
    \item Compute $\alpha = \max\{0,\min\{1,1/2-x_0\}\}$
    \item Output $u = \alpha w^- + (1-\alpha) w^+$
\end{enumerate}
where 
$$\left[\widehat{b}_i(x)\right]_j = \left\{ \begin{tabular}{ll}
    $x_i(1-x_0)$ & if $j = 0$ \\
    $1-x_0-x_i^2$ & if $j = i$ \\
    $-x_j x_i$ & otherwise
\end{tabular} \right.$$
is the non-normalised version of the tangent basis used in the proof of \cref{thm:HBinPPAD}. In particular, the $\widehat{b}_i(x)$ form an orthogonal basis of the tangent space at $x$ (as long as $x_0 < 1$) and we have $\|\widehat{b}_i(x)\|_2 = 1-x_0$. Note that the output $u$ of the circuit is always tangent to the sphere at $x$.

If $x_0 \leq -1/2$, then $u=w^-$, and thus, using the fact that the $\widehat{b}_i(x)$ form an orthogonal basis, we obtain
$$\|u\|_2^2 = \left\langle \sum_{i=1}^k v_i \widehat{b}_i(x), \sum_{i=1}^k v_i \widehat{b}_i(x) \right\rangle = \sum_{i=1}^k v_i^2 \left\langle \widehat{b}_i(x), \widehat{b}_i(x) \right\rangle = \sum_{i=1}^k v_i^2 \|\widehat{b}_i(x)\|_2^2 = \|v\|_2^2 (1-x_0)^2.$$
Since $1-x_0 \geq 1$, this implies that $\|u\|_2 \geq \|v\|_2$, i.e.\ $\|F(x_0,x_1,\dots, x_k)\|_2 \geq \|\widehat{G}(x_1,\dots, x_k)\|_2$. This means that any zero of $F$ yields a zero of $\widehat{G}$ by dropping the first coordinate.

If $x_0 \geq 1/2$, then $u=w^+$, and the analogous computation yields that $\|u\|_2^2 = \sum_{i=1}^k v_i^2 \|\widehat{b}_i(-x)\|_2^2 = \|v\|_2^2 (1+x_0)^2$. Since $1+x_0 \geq 1$, we again obtain that $\|u\|_2 \geq \|v\|_2$, and the same conclusion follows.

The last remaining case is when $|x_0| < 1/2$. If we let $x'=(x_1, \dots, x_k)$, then, since $\|x\|_2 = 1$, it must be that $\|x'\|_2 \geq 1/2$, which implies that $v = \widehat{G}(x') = -x'$. It follows that
\begin{equation*}
\begin{split}
u_0 = \alpha w^-_0 + (1-\alpha)w^+_0 &= \alpha \sum_{i=1}^k v_i \left[\widehat{b}_i(x)\right]_0 + (1-\alpha) \sum_{i=1}^k (-v_i) \left[\widehat{b}_i(-x)\right]_0\\
&= \alpha \sum_{i=1}^k (-x_i) x_i(1-x_0) + (1-\alpha) \sum_{i=1}^k x_i (-x_i) (1+x_0)\\
&=\alpha \sum_{i=1}^k (-x_i^2) (1-x_0) + (1-\alpha)  \sum_{i=1}^k (-x_i^2) (1+x_0)\\
&= -\|x'\|_2^2 (1 + x_0 (1-2\alpha)) \leq -\|x'\|_2^2/2 \leq -1/8
\end{split}
\end{equation*}
which implies that $F$ has no zero in this region.

The reduction we have provided is an SL-reduction, since a solution of the original instance can be obtained by applying a separable linear transformation (with rational coefficients that are computable in polynomial time from the original instance) to any solution of the final instance. Furthermore, the coefficients of the separable linear transformation are always powers of $2$. It follows (see~\cite{etessami2010fixp}) that the corresponding strong approximation, partial computation and decision problems are also hard for the corresponding FIXP classes.
\end{proof}

\section{The 2D-Hairy-Cube Problem}\label{sec:hairycube}

In this section we define the \textsc{2D-Hairy-Cube} problem: a discretised version of \textsc{2D-Hairy-Ball}. As the name suggests, one of the main differences is that \textsc{2D-Hairy-Cube} is defined on the surface of a cube instead of the surface of a ball. It will follow from our results that this discretised problem is indeed computationally equivalent to the continuous formulation.

We consider this problem mainly for two reasons. First of all, we believe it is nice to have a corresponding discrete version of \textsc{2D-Hairy-Ball} that might easier to work with in some cases. This is similar to \textsc{Brouwer}, where the discrete version is often very useful. Another more technical point is that \textsc{2D-Hairy-Cube} also serves as an intermediate step to prove PPAD-hardness of \textsc{2D-Hairy-Ball}. Indeed, as we will see below, proving PPAD-hardness for \textsc{2D-Hairy-Cube} is pretty straightforward using a result by Chen and Deng~\cite{chen2009complexity} on a 2-dimensional Brouwer fixpoint problem.

\subsection{Definition}

The \textsc{2D-Hairy-Cube} problem is defined on the surface of the unit cube. Every face of the cube is subdivided into $2^n \times 2^n$ identical small squares, that we call squarelets. Every squarelet contains one of four possible axis-aligned arrows tangential to its face. A solution is given by two adjacent squarelets that contain arrows going into opposite directions. Here adjacent is taken to mean that the squarelets touch, even in a single point. However, we need one additional type of solution, namely at corners. Three squarelets adjacent at a corner form a solution, if the three arrows all point towards the corner or all three point away from the corner. An arrow points towards the corner if it points towards an edge that contains the corner. Thus, for any corner squarelet, out of the four possible arrows that it can contain, two point towards the corner and two point away from it.

Note that any two adjacent faces of the cube can be thought of as being one big plane, because we can rotate one of the faces so that this is the case. This is the interpretation that is used when deciding whether two adjacent squarelets, lying on two different but adjacent faces, form a solution: we rotate one of the two faces until it is on the same plane as the other one, and then proceed as if the two squarelets were on the same face. \cref{fig:Hairy-Cube} shows a \textsc{2D-Hairy-Cube} instance with various solutions.

\newcommand{\DrawUpArrow}{
\begin{tikzpicture}
\draw[-{Latex[length=2mm,width=2mm]},thick] (0,0) -- (0,0.5);
\end{tikzpicture}
}
\newcommand{\DrawDownArrow}{
\begin{tikzpicture}
\draw[-{Latex[length=2mm,width=2mm]},thick] (0,0) -- (0,-0.5);
\end{tikzpicture}
}
\newcommand{\DrawRightArrow}{
\begin{tikzpicture}
\draw[-{Latex[length=2mm,width=2mm]},thick] (0,0) -- (0.5,0);
\end{tikzpicture}
}
\newcommand{\DrawLeftArrow}{
\begin{tikzpicture}
\draw[-{Latex[length=2mm,width=2mm]},thick] (0,0) -- (-0.5,0);
\end{tikzpicture}
}
\begin{figure}[h]
\centering
\begin{tikzpicture}
\begin{scope}[every node/.append style={yslant=-0.5}, yslant=-0.5]
  \fill[color=black, opacity=0.2] (0,0) rectangle (3,3);
  \node at (0.5,2.5) {\DrawUpArrow};
  \node at (1.5,2.5) {\DrawUpArrow};
  \node at (2.5,2.5) {\DrawUpArrow};
  \node at (0.5,1.5) {\DrawLeftArrow};
  \node at (1.5,1.5) {\DrawLeftArrow};
  \node at (2.5,1.5) {\DrawLeftArrow};
  \node at (0.5,0.5) {\DrawDownArrow};
  \node at (1.5,0.5) {\DrawUpArrow};
  \node at (2.5,0.5) {\DrawLeftArrow};
  \draw (0,0) grid (3,3);
  \draw[very thick] (0,3) -- (3,3);
  \draw[very thick] (3,0) -- (3,3);
  \draw[thick,red] (0,0) rectangle (2,1);
  \draw[thick,red] (2,3) -- (2,2) -- (3,2);
\end{scope}
\begin{scope}[every node/.append style={yslant=0.5}, yslant=0.5]
  \fill[color=black, opacity=0.15] (3,-3) rectangle (6,0);
  \node at (3.5,-0.5) {\DrawLeftArrow};
  \node at (4.5,-0.5) {\DrawLeftArrow};
  \node at (5.5,-0.5) {\DrawUpArrow};
  \node at (3.5,-1.5) {\DrawUpArrow};
  \node at (4.5,-1.5) {\DrawUpArrow};
  \node at (5.5,-1.5) {\DrawUpArrow};
  \node at (3.5,-2.5) {\DrawUpArrow};
  \node at (4.5,-2.5) {\DrawLeftArrow};
  \node at (5.5,-2.5) {\DrawLeftArrow};
  \draw (3,-3) grid (6,0);
  \draw[very thick] (3,0) -- (6,0);
  \draw[thick,red] (5,0) -- (5,-1) -- (6,-1) -- (6,0);
  \draw[thick,red] (3,-1) -- (4,-1) -- (4,0);
\end{scope}
\begin{scope}[every node/.append style={yslant=0.5,xslant=-1}, yslant=0.5, xslant=-1]
  \fill[color=black, opacity=0.1] (3,0) rectangle (6,3);
  \node at (3.5,2.5) {\DrawRightArrow};
  \node at (3.5,1.5) {\DrawDownArrow};
  \node at (3.5,0.5) {\DrawDownArrow};
  \node at (4.5,2.5) {\DrawDownArrow};
  \node at (4.5,1.5) {\DrawDownArrow};
  \node at (4.5,0.5) {\DrawLeftArrow};
  \node at (5.5,2.5) {\DrawLeftArrow};
  \node at (5.5,1.5) {\DrawDownArrow};
  \node at (5.5,0.5) {\DrawDownArrow};
  \draw (3,0) grid (6,3);
  \draw[thick,red] (5,0) -- (5,1) -- (6,1) -- (6,0);
  \draw[thick,red] (3,1) -- (4,1) -- (4,0);
\end{scope}
\end{tikzpicture}
\caption{Partial view of a \textsc{2D-Hairy-Cube} instance with $3 \times 3$ squarelets per face (instead of $2^n \times 2^n$). Three solutions are visible from this point of view: one at the central corner, one on the left face and one where the right and top face meet.}
\label{fig:Hairy-Cube}
\end{figure}

For any integer $m \geq 1$, we let $[m] := \{1, 2, \dots, m\}$. Note that we can interpret elements in $\{0,1\}^n$ as elements in $[2^n]$.

For the definition of the computational problem we assume that we are given a Boolean circuit $C: \{0,1\}^n \times \{0,1\}^n \times \{0,1\}^3 \to \{0,1\}^2$ that represents a function $[2^n] \times [2^n] \times [6] \to [4]$. The input in $\{0,1\}^3$ is interpreted as an element in $\{1,2,\dots,6\}$ that determines a face of the cube. The input in $\{0,1\}^n \times \{0,1\}^n$ determines a squarelet on that face. The output in $\{0,1\}^2$ is interpreted as an element in $\{1,2,3,4\}$, which yields one of the four arrows allowed on this face. We say that $C$ computes a discrete field on the cube surface.

\begin{definition}[\textsc{2D-Hairy-Cube}]
The \textsc{2D-Hairy-Cube} problem is defined as: given a Boolean circuit $C: \{0,1\}^n \times \{0,1\}^n \times \{0,1\}^3 \to \{0,1\}^2$ that computes a discrete field on the cube surface, find:
\begin{itemize}
    \item two adjacent squarelets that contain arrows going in opposite directions, or
    \item three adjacent corner squarelets, such that all three arrows point towards the corner, or all three point away from the corner.
\end{itemize}
\end{definition}

In \cref{sec:app:cubetoHB} we prove that this problem reduces to \textsc{2D-Hairy-Ball} (\cref{prop:cubetoHB}). Thus, it also lies in PPAD by \cref{cor:linearHBinPPAD}. Below we prove that it is also PPAD-hard (\cref{prop:cubePPADhard}), which yields:

\begin{theorem}
\textsc{2D-Hairy-Cube} is \textup{PPAD}-complete.
\end{theorem}

\subsection{PPAD-hardness}\label{sec:brouwer-ppad-hard}

In this section we prove that \textsc{2D-Hairy-Cube} is PPAD-hard. The proof uses the PPAD-hardness of the discrete \textsc{2D-Brouwer} problem~\cite{chen2009complexity} by embedding an instance of a slightly modified discrete \textsc{2D-Brouwer} problem on two opposite faces of the cube. We start by defining the modified discrete \textsc{2D-Brouwer} problem and proving that it is PPAD-hard.

The 2-dimensional discrete Brouwer fixpoint problem we will consider is defined on a big square that is partitioned into $2^n \times 2^n$ squarelets. Every squarelet contains one of four possible cardinal direction arrows: right, left, up, down. A solution is given by two adjacent squarelets that contain arrows pointing in opposite directions. Two squarelets are adjacent if they have at least one point in common (in particular, squarelets can also be adjacent diagonally).

The input is a Boolean circuit $C: [2^n] \times [2^n] \to \{(1,0),(-1,0),(0,1),(0,-1)\}$ that computes the cardinal arrow in a given input squarelet. We say that $C$ \emph{satisfies the Brouwer boundary conditions}, if $C(1,a)=(1,0)$, $C(2^n,b)=(-1,0)$, $C(b,1)=(0,1)$ and $C(a,2^n)=(0,-1)$, for all $a \in [2^n] \setminus \{1\}$ and $b \in [2^n-1]$. This corresponds to requiring that, in squarelets forming the boundary of the domain, the arrows point inwards. Note that this condition can be enforced syntactically on the circuit. If $C$ satisfies the Brouwer boundary conditions, then there must exist a solution. Otherwise, we could use the instance to define a continuous function from the square to itself which does not have a fixpoint, a contradiction to Brouwer's Fixpoint theorem. We now give the formal definition of the problem.

\begin{definition}[\textsc{2D-Variant-Brouwer}]
The problem \textsc{2D-Variant-Brouwer} is defined as: given a Boolean circuit $C: [2^n] \times [2^n] \to \{(1,0),(-1,0),(0,1),(0,-1)\}$ such that $C(1,a)=(1,0)$, $C(2^n,b)=(-1,0)$, $C(b,1)=(0,1)$ and $C(a,2^n)=(0,-1)$, for all $a \in [2^n] \setminus \{1\}$ and $b \in [2^n-1]$, find $x$ and $y$ such that $\|x-y\|_\infty \leq 1$ and $C(x)+C(y)=0$.
\end{definition}

The main difference to the standard discrete \textsc{2D-Brouwer} problem~\cite{chen2009complexity} is that we have $4$ possible arrows instead of the three $(1,0),(0,1),(-1,-1)$ and we are looking for a place where two opposite arrows meet, instead of a place where all three types of arrows meet. Any \textsc{2D-Variant-Brouwer} instance can also be interpreted as a \textsc{2D-Tucker} instance~\cite{aisenberg20152}. Indeed, if we interpret $(1,0),(-1,0),(0,1),(0,-1)$ as $+1,-1,+2,-2$, then the boundary satisfies the Tucker conditions (i.e.\ two boundary squarelets that are diametrically opposite each other with respect to the centre of the domain, contain opposite numbers). Furthermore, \textsc{2D-Variant-Brouwer}-solutions exactly correspond to \textsc{2D-Tucker}-solutions. Thus, one might think that \textsc{2D-Variant-Brouwer} could be PPA-complete instead of PPAD-complete. However, the corresponding \textsc{2D-Tucker} instances have the specificity of having monochromatic sides (except for the corners) and this puts the problem in PPAD. Indeed, the proof of PPA-hardness of \textsc{2D-Tucker} uses instances with non-monochromatic sides~\cite{aisenberg20152}.

\begin{lemma}\label{lem:2dBrwrPPAD}
\textsc{2D-Variant-Brouwer} is \textup{PPAD}-complete.
\end{lemma}

\begin{proof}
Membership in PPAD follows from the fact that \textsc{2D-Variant-Brouwer} reduces to \textsc{2D-Hairy-Cube} (\cref{prop:cubePPADhard}), which reduces to \textsc{2D-Hairy-Ball} (\cref{prop:cubetoHB}) which is in PPAD (\cref{cor:linearHBinPPAD}). Alternatively, one could also reduce \textsc{2D-Variant-Brouwer} to the continuous \textsc{2D-Brouwer} problem (by interpolating between centres of the squarelets) which takes a continuous function as input (represented as an arithmetic circuit) and is known to lie in PPAD~\cite{papadimitriou1994complexity}.

It remains to show that \textsc{2D-Variant-Brouwer} is PPAD-hard. This is done by using the same idea as for the standard \textsc{2D-Brouwer} problem~\cite{chen2009complexity}, namely by embedding a 2-dimensional \textsc{End-of-Line} instance on the grid into a \textsc{2D-Variant-Brouwer} instance. Since the 2-dimensional \textsc{End-of-Line} problem is PPAD-hard~\cite{chen2009complexity}, the result follows.

The only difference with \cite{chen2009complexity} is that in our case, paths are implemented using an ``upward'' arrow at the centre of the path, a ``right'' arrow on the left side of the path, and a ``left'' arrow on the right side of the path. The environment, where there is no path, is filled with ``downward'' arrows. As a result, a solution, i.e., a point where two opposite arrows meet, can only occur when a path starts or stops in the environment. The same type of construction for the paths (if we interpret labels as arrows) was used to show that \textsc{2D-Tucker} is PPAD-hard~\cite{palvolgyi20092d}. The correctness of the reduction follows by the arguments presented in \cite{chen2009complexity}.

We provide a few more details about the construction for the interested reader. Consider any 2-dimensional \textsc{End-of-Line} instance on the $[2^n] \times [2^n]$ grid. We will construct a \textsc{2D-Variant-Brouwer} instance on the $[2^{n+3}] \times [2^{n+3}]$ grid. Every grid node in the original \textsc{End-of-Line} instance corresponds to a region of $8 \times 8$ squarelets. In this region, the \textsc{End-of-Line} path is represented by squarelets containing the ``upward'' arrow $(0,1)$. Thus, if the path comes into the grid node from the left neighbour and leaves towards the bottom neighbour, then the $8 \times 8$ region will contain a width-2 path of squarelets that starts at the centre of the left side of the region, makes a right angle at the centre of the region and then goes to the centre of the bottom side of the region. All the squarelets on this width-$2$ path get the arrow $(0,1)$. Any squarelets that are adjacent to the path from the left side (with respect to the direction of the path) get the arrow $(1,0)$. Any squarelets that are adjacent to the path from the right side get the arrow $(-1,0)$. All other squarelets get the arrow $(0,-1)$.

If a grid node is isolated, then we just assign the arrow $(0,-1)$ to all of its squarelets. If a grid node is a source or a sink, then there will be a solution, since the path of $(0,1)$ arrows will touch the environment of $(0,-1)$ arrows. The \textsc{2D-Variant-Brouwer} boundary conditions can be enforced without changing the solutions. We just need to be careful at the bottom side, since the boundary has arrow $(0,1)$ and we don't want them to touch the environment of $(0,-1)$ arrows. To avoid this we also require that the stripe of width $1$ just above the bottom boundary be filled with $(-1,0)$ arrows. It is easy to see that a careful construction of the region corresponding to the origin of the grid (the known source of the \textsc{End-of-Line} instance) ensures that no solution occurs there. It follows that any solution to the \textsc{2D-Variant-Brouwer} instance must yield a solution to the 2-dimensional \textsc{End-of-Line} instance. Furthermore, the circuit $C$ computing the arrow in each squarelet can be constructed efficiently.
\end{proof}

\begin{proposition}\label{prop:cubePPADhard}
\textsc{2D-Hairy-Cube} is \textup{PPAD}-hard.
\end{proposition}

\begin{proof}
We reduce from \textsc{2D-Variant-Brouwer}. By \cref{lem:2dBrwrPPAD} the result then follows. Consider an instance $C$ of \textsc{2D-Variant-Brouwer}. We construct an instance of \textsc{2D-Hairy-Cube} as follows. The top face of the cube is identical to the instance represented by $C$. The bottom face is identical to the negation of the instance represented by $C$, i.e.\ all arrows are negated (in particular, on the boundary they point outwards, instead of inwards). The other four faces of the cube are uniformly filled with an arrow pointing towards the top face. It is easy to check that a solution can only appear on the top or bottom face and corresponds to a solution to the original problem. Furthermore, the circuit computing the arrow in any squarelet of the surface of the cube can be constructed efficiently from $C$.
\end{proof}

\section{The Hairy-Ball Problem on the Torus of Genus \texorpdfstring{$g$}{g}}\label{sec:gTorus}

In this section we consider the 2-dimensional torus of genus $g$ (also called $g$-holed torus, or $g$-torus) instead of the sphere. The Hairy Ball theorem does not hold on the torus of genus $1$, i.e.\ the standard torus. It is straightforward to define a continuous tangent field that never vanishes. However, as proved by Poincar{\'e}~\cite{poincare1885courbes}, the theorem does hold for the 2-dimensional torus of genus $g$ for all $g \geq 2$.

We prove that the corresponding computational problem is again PPAD-complete. Even though the general techniques are very similar to those seen in previous sections, the interesting point is that the reduction to \textsc{End-of-Line} yields $2(g-1)$ sources instead of $2$. Thus, we can make full use of the multiple-source \textsc{End-of-Line} results of \cref{sec:EOL}. We believe that the reduction to a $2(g-1)$-source \textsc{End-of-Line} instance is a natural and necessary step in the reduction and that this multiple source aspect is intrinsic to the problem, since the torus of genus $g$ has Euler characteristic $2-2g$~\cite{guillemin1974differential}.

\subsection{The \texorpdfstring{$g$}{g}-torus and the Hairy-\texorpdfstring{$g$}{g}-Torus problem}

We use the term ``$g$-torus'' to refer to the torus of genus $g$, i.e.\ the 2-dimensional torus with $g$ holes, embedded in $3$-dimensional space. This should not be confused with the 1-holed torus embedded in $g$-dimensional space. We mainly work with the 2-torus, since all the arguments generalise to the $g$-torus in a straightforward way. We define the 2-torus as follows: $(x_1,x_2,x_3)$ lies on the 2-torus if it satisfies the equation
$$x_3^2+\min\{(d_1-2)^2,(d_2-2)^2,(d_3-2)^2, \max\{0,2-d_1,2-d_2,d_3-2\}\} = 1$$
where $d_1$ is the distance of $(x_1,x_2,0)$ to $(0,0,0)$, $d_2$ is the distance to $(5,0,0)$, and $d_3$ is the distance to the segment $(0,0,0)$-$(5,0,0)$. \cref{fig:2-torus} provides an illustration of this 2-torus.

Using the expression above, one can always efficiently check whether a given $(x_1,x_2,x_3)$ lies on the 2-torus. However, it is not always possible to exactly compute the projection of a vector onto the tangent space at some point on the 2-torus.

This construction is easy to generalise to the $g$-hole torus, for any $g \geq 2$. We call the corresponding object the $g$-torus and denote it by $T_g$.

Both of our reductions will use a \emph{base field} $F_B$, which is defined everywhere except on the two flat regions of the 2-torus. If $x_1^2+x_2^2 \in [1,2]$ or if $x_1^2+x_2^2 \in [2,3]$ and $x_1 \leq 0$, then $F_B(x_1,x_2,x_3)=(x_2,-x_1,0)$. If $(x_1-5)^2+x_2^2 \in [1,2]$ or if $(x_1-5)^2+x_2^2 \in [2,3]$ and $x_1 \geq 5$, then $F_B(x_1,x_2,x_3)=(x_2,-x_1+5,0)$. If $x_1 \in [0,5]$ and $|x_2| \in [2,3]$, then $F_B(x_1,x_2,x_3)=(x_2,0,0)$. This partial field is continuous and does not have any zeros. Note that the field does not depend on $x_3$ and is always contained in the $(x_1,x_2)$-plane.

\begin{figure}[h]
\centering
\begin{tikzpicture}
\definecolor{nicegray}{rgb}{0.9,0.9,0.9}
\pgfdeclareradialshading{bwring}{\pgfpoint{0cm}{0cm}}
{
color(0cm)=(gray);
color(1cm)=(gray);
rgb(2cm)=(0.9,0.9,0.9);
color(4cm)=(black)
}
\begin{scope}
\clip (0,2) arc(90:-90:2) -- (0,-3) arc(270:90:3) --cycle;
\pgfuseshading{bwring}
\fill[color=white] circle (1);
\end{scope}

\begin{scope}
\clip (5,2) arc(90:270:2) -- (5,-3) arc(-90:90:3) --cycle;
\pgfputat{\pgfxy(5,0)}{\pgfbox[center,center]{\pgfuseshading{bwring}}}
\fill[color=white] (5,0) circle (1);
\end{scope}

\pgfdeclareverticalshading{bwrectangle}{5.02cm}
{
rgb(0cm)=(0.9,0.9,0.9);
rgb(1cm)=(0.45,0.45,0.45)
}
\pgfputat{\pgfxy(2.5,2.5)}{\pgfbox[center,center]{\pgfuseshading{bwrectangle}}}
\pgfputat{\pgfxy(2.5,-2.5)}{\begin{pgfrotateby}{\pgfdegree{180}}\pgfbox[center,center]{\pgfuseshading{bwrectangle}}\end{pgfrotateby}}

\path[fill=nicegray] (5,2.01) -- (5,2) arc(90:270:2) -- (5,-2.01) -- (0,-2.01) -- (0,-2) arc(-90:90:2) -- (0,2.01) --cycle;

\draw[thick] (0,0) circle (1cm);
\draw[thick] (0,3) arc(90:270:3);
\draw[thick] (5,0) circle (1cm);
\draw[thick] (5,-3) arc(-90:90:3);
\draw[thick] (0,3) -- (5,3);
\draw[thick] (0,-3) -- (5,-3);

\draw[dotted] (0,2) -- (5,2) arc(90:270:2) -- (0,-2) arc(-90:90:2);

\draw[->,thick] (1.1,2.5) -- ++(0.7,0);
\draw[->,thick] (3.4,2.5) -- ++(0.7,0);
\draw[->,thick] (3.9,-2.5) -- ++(-0.7,0);
\draw[->,thick] (1.6,-2.5) -- ++(-0.7,0);

\draw[->,thick] (1.06,1.06) -- (1.41,0.71);
\draw[->,thick] (1.06,-1.06) -- (0.71, -1.41);
\draw[->,thick] (-1.06,-1.06) -- (-1.41,-0.71);
\draw[->,thick] (-1.06,1.06) -- (-0.71,1.41);
\draw[->,thick] (6.06,1.06) -- (6.41,0.71);
\draw[->,thick] (6.06,-1.06) -- (5.71,-1.41);
\draw[->,thick] (3.94,-1.06) -- (3.59,-0.71);
\draw[->,thick] (3.94,1.06) -- (4.29,1.41);

\draw[->,thick] (-2.5,0) -- (-2.5,0.7);
\draw[->,thick] (7.5,0) -- (7.5,-0.7);

\draw[->,thick] (-1.25,-2.165) -- (-1.86,-1.815);
\draw[->,thick] (-1.25,2.165) -- (-0.64,2.515);
\draw[->,thick] (6.25,-2.165) -- (5.64,-2.515);
\draw[->,thick] (6.25,2.165) -- (6.86,1.815);

\end{tikzpicture}
\caption{The $2$-torus and the tangent continuous field $F_B$, observed from above. Note that the base field is not defined on the flat region at the centre, which corresponds to the area delimited by the dotted line in this figure.}
\label{fig:2-torus}
\end{figure}

The computational problem is defined as follows:
\begin{definition}[\textsc{Hairy-$g$-Torus}]
Let $g \geq 2$. The \textsc{Hairy-$g$-Torus} problem is defined as: given $\varepsilon > 0$ and an arithmetic circuit $G$ with 3 inputs and outputs, using gates $\{+, \times \zeta, \max\}$ and rational constants, find $x \in T_g$ such that $\|P_x^\varepsilon[G(x)]\|_\infty \leq \varepsilon$.
\end{definition}
Here $P_x^\varepsilon$ denotes the projection onto the tangent space to $T_g$ at $x \in T_g$ with error at most $\varepsilon/2$, i.e. $\|P_x[G(x)]-P_x^\varepsilon[G(x)]\|_\infty \leq \varepsilon/2$, where $P_x$ is the exact projection.

\begin{theorem}\label{thm:gTorus-ppad-complete}
For any $g \geq 2$, \textsc{Hairy-$g$-Torus} is \textup{PPAD}-complete.
\end{theorem}

\begin{proof}
In \cref{prop:gTorus-to-EOL} below, we show that \textsc{Hairy-$g$-Torus} reduces to $2(g-1)$-source \textsc{End-of-Line}. Using the results in \cref{sec:EOL} (\cref{thm:mulsourceeol}), it follows that the problem lies in PPAD. PPAD-hardness is proved in \cref{sec:app:gTorus-ppad-hard} by embedding $2(g-1)$ modified copies of a \textsc{2D-Variant-Brouwer} instance.
\end{proof}

\subsection{Hairy-\texorpdfstring{$g$}{g}-Torus is in PPAD}

\begin{proposition}\label{prop:gTorus-to-EOL}
For any $g \geq 2$, \textsc{Hairy-$g$-Torus} reduces to $2(g-1)$-source \textsc{End-of-Line}.
\end{proposition}

\begin{proof}
We will prove the result for $g=2$. The proof immediately generalises to any other larger value of $g$.

For this proof we need a sufficiently fine triangulation of the $2$-torus that is also locally efficiently computable. For any point on the $2$-torus, we should be able to efficiently compute all nearby triangles. In particular, given a triangle we should be able to efficiently obtain all adjacent triangles. Moreover, the triangulation should be sufficiently fine, meaning that, given some $\delta > 0$, we should have a triangulation such that all corners of a triangle are within distance $\delta$ of each other. We call this a $\delta$-fine triangulation.

We can construct such a triangulation as follows. We start with a fixed triangulation of the $2$-torus, such that the curved region delimited by each triangle is sufficiently flat. Then, we can easily construct a fine triangulation of every triangle. Finally, we compute the (approximate) projection of the points of this fine triangulation on the surface of the $2$-torus. By taking a fine enough triangulation of the triangles and a good enough approximation of the projection, we can ensure that the triangulation is $\delta$-fine. Note that when we take the approximate projection, we compute a point that \emph{exactly} lies on the $2$-torus and is close to the true projection. Thus, all the points of the triangulation lie exactly on $T_2$.

Let $(\varepsilon,G)$ be an instance of \textsc{Hairy-$2$-Torus}. Using similar arguments to the proof of \cref{lem:linearHBlipschitz}, it is straightforward to show that $x \to P_x[G(x)]$ is $L$-Lipschitz continuous on the $2$-torus, for a similar value of $L$. Pick a $\delta$-fine triangulation of the $2$-torus as described above with $\delta = \varepsilon \sqrt{2(1-\cos 1^\circ)} / (2L)$. We will now colour every node of the triangulation with one of three colours $\{A,B,C\}$ and use a Sperner argument on this colouring. To construct the colouring we will use the \emph{base field} $F_B$, defined earlier. We just need to decide how to extend the base field to also be defined on the two flat regions (at $x_3=1$ and $x_3=-1$). Here, the simplest way to do this is to draw an X-shape on the flat region and assign one of the four cardinal flat vectors to each of the four regions (see \cref{fig:2-torus-flat}). Then, the base field is defined for any point on $T_2$ (if the point lies on a boundary, resolve ties arbitrarily), even though it is no longer continuous.

\begin{figure}[h]
\centering
\begin{tikzpicture}[scale=2]
\definecolor{nicegray}{rgb}{0.9,0.9,0.9}
\pgfdeclareradialshading{smallbwring}{\pgfpoint{0cm}{0cm}}
{
color(0cm)=(gray);
color(1cm)=(gray);
rgb(2cm)=(0.9,0.9,0.9);
color(4cm)=(black)
}
\begin{scope}
\clip (0,2) arc(90:-90:2) -- (0,-1) arc(-90:90:1) --cycle;
\begin{pgfmagnify}{2}{2}
\pgfuseshading{smallbwring}
\end{pgfmagnify}
\end{scope}
\begin{scope}
\clip (5,2) arc(90:270:2) -- (5,-1) arc(270:90:1) --cycle;
\begin{pgfmagnify}{2}{2}
\pgfputat{\pgfxy(5,0)}{\pgfbox[center,center]{\pgfuseshading{smallbwring}}}
\end{pgfmagnify}
\end{scope}

\path[fill=nicegray] (5,2) arc(90:270:2) -- (0,-2) arc(-90:90:2) --cycle;

\draw[thick,dotted] (1.73,1) -- (3.27,-1);
\draw[thick,dotted] (1.73,-1) -- (3.27,1);

\draw[thick] (0,1) arc(90:-90:1);
\draw[thick] (5,1) arc(90:270:1);
\draw[dotted] (0,2) arc(90:-90:2);
\draw[dotted] (5,2) arc(90:270:2);
\draw[dotted] (0,2) -- (5,2);
\draw[dotted] (0,-2) -- (5,-2);

\draw[->,thick] (2.85,-0.25) --++ (0,0.5);
\draw[->,thick] (2.15,0.25) --++ (0,-0.5);
\draw[->,thick] (2.25,1.25) --++ (0.5,0);
\draw[->,thick] (2.75,-1.25) --++ (-0.5,0);

\draw[->] (1.5,0) -- (1.5,-0.5);
\draw[->] (0.2,1.47) --++ (0.51,-0.08);
\draw[->] (1.06,1.06) -- (1.41,0.71);
\draw[->] (1.06,-1.06) -- (0.71, -1.41);

\draw[->] (3.5,0) --++ (0,0.5);
\draw[->] (4.8,-1.47) --++ (-0.51,0.08);
\draw[->] (3.94,-1.06) -- (3.59,-0.71);
\draw[->] (3.94,1.06) -- (4.29,1.41);

\end{tikzpicture}
\caption{One of the two flat regions of the 2-torus. This figure illustrates how the base field is extended to be defined on the flat regions, using only the four cardinal directions.}
\label{fig:2-torus-flat}
\end{figure}

For any node $x$ of the triangulation, its colour depends on $\theta \in [0^\circ,360^\circ)$ : the angle between $P_x[G(x)]$ and the base field at $x$. The angle is measured from $P_x[G(x)]$ to the base field vector in counter-clockwise direction when looking from outside the $2$-torus. If $0^\circ \leq \theta < 120^\circ$ then the colour is $A$, if $120^\circ \leq \theta < 240^\circ$ it is $B$, and if $240^\circ \leq \theta < 360^\circ$ it is $C$. Unfortunately, we cannot exactly compute $P_x[G(x)]$, nor the angle between two vectors. Thus, we compute a sufficiently good approximation of $P_x[G(x)]$ and of the angle and choose the colour based on this approximation. Specifically, we compute a good enough approximation such that if $\|P_x[G(x)]\|_\infty \geq \varepsilon/2$, then the angle has error at most $1^\circ$.

We can assume that at the centre of the cross in the top flat region ($x_3=1$), the triangulation is such that there are four points such that each lies in one of the four regions, and they form two triangles of the triangulation. Let $x^{(1)},x^{(2)},x^{(3)},x^{(4)}$ be those four points. Note that the pairwise distance between any of those four points is at most $2 \delta$. Thus, by the choice of $\delta$, we must have that $\|P_{x^{(i)}}[G(x^{(i)})] - P_{x^{(j)}}[G(x^{(j)})]\|_\infty \leq \varepsilon \sqrt{2(1-\cos 1^\circ)}$, for any $i,j$. Note that since we are in the flat region, we can project exactly, i.e.\ $P_x^\varepsilon[G(x)] = P_x[G(x)]$. If there exists $i$ such that $\|P_{x^{(i)}}[G(x^{(i)})]\|_\infty \leq \varepsilon$, then we are done. If this is not the case, then the angle of $P_{x^{(i)}}[G(x^{(i)})]$ varies by at most one degree for $i \in \{1,2,3,4\}$. It follows that the angle $\theta_i$ computed for each of the four points is equal to the angle between the base field and some fixed vector, with an error of at most $1^\circ$. Since the four points' base vectors correspond to the four cardinal flat vectors, it is easy to check that the four points get the colours $A,B,C$ in clockwise order (and one of the colours is repeated, i.e.\ $A,A,B,C$ or $A,B,B,C$ or $A,B,C,C$). The $AB$ segment is our first Sperner-source. Applying the same arguments to the bottom flat region ($x_3=-1$) yields four points that again get the colours $A,B,C$ in clockwise order. Thus, we get a second $AB$ segment that is another Sperner-source. Each of those sources is the beginning of a Sperner path that will eventually lead to a trichromatic triangle. Since the $2$-torus is orientable, the two paths cannot cancel each other out.

It remains to show that any trichromatic triangle yields a solution to the \textsc{Hairy-$2$-Torus} instance. Consider any trichromatic triangle with three vertices $x^{(1)},x^{(2)},x^{(3)}$. Assuming we have picked $\delta$ smaller than some constant, the base field can vary at most by an angle of $90^\circ$ on these three vertices. If $\|P_{x^{(i)}}^\varepsilon[G(x^{(i)})]\|_\infty \leq \varepsilon$ for some $i$, then we are done. Assume that this is not the case. For all $i,j$ we also have that $\|P_{x^{(i)}}[G(x^{(i)})] - P_{x^{(j)}}[G(x^{(j)})]\|_\infty \leq L \delta \leq \varepsilon \sqrt{2(1-\cos 1^\circ)}/2$ and thus the angle between those two vectors is at most $1^\circ$ (because $\|P_{x^{(i)}}[G(x^{(i)})]\|_\infty > \varepsilon/2$). Thus, the $\theta_i$ can vary by at most $93^\circ$ for $i=1,2,3$ (because every $\theta_i$ is computed with error at most $1^\circ$ in this case). It follows that it is impossible to get all three colours.

For larger $g$, we define the base field by just repeating the same pattern and the rest works as before. Note that we obtain two sources for every region between two adjacent holes, i.e.\ $2(g-1)$ sources.
\end{proof}

\section{End-of-Line: One Source to Rule Them All}\label{sec:EOL}

\textsc{End-of-Line} is the canonical problem used to define PPAD. Investigating variants of the problem is of independent interest, in particular in order to gain a better understanding of PPAD and how it relates to other similar subclasses of TFNP. An additional motivation for studying these variants is given by this paper, since a multiple-source variant of \textsc{End-of-Line} is used to prove that finding an approximate Hairy Ball solution lies in PPAD (\cref{sec:HBinPPAD}).

This section is an improved version of the corresponding content in the technical report~\cite{HG2018multsource}. In~\cite{HG2018multsource} we use these results to show that a computational problem related to the Mutilated Chessboard puzzle is PPAD-complete.

\subsection{Multiple-Source End-of-Line}

Recall that in the \textsc{End-of-Line} problem (\cref{def:endofline}), we are given a directed graph where each vertex has in- and out-degree at most $1$ and a known source of this graph, and we wish to find a sink or another source. But, what if, instead of just one, we already know \emph{two} sources of an \textsc{End-of-Line} instance? We are still interested in finding any sink or any \emph{other} source. Intuitively, the problem might seem easier, because the existence of two sources implies the existence of at least two sinks, hence more potential solutions. In fact, it is easy to see that this problem is actually at least as hard as \textsc{End-of-Line}: just duplicate the whole \textsc{End-of-Line} instance.

The other direction, however, is not trivial. Indeed, if we interpret our 2-source \textsc{End-of-Line} instance  as a standard \textsc{End-of-Line} instance (and pick one of the two sources as the standard source), then the other known source is a valid solution to \textsc{End-of-Line}, but not a valid solution to our original problem. In other words, it is not clear how to solve this problem if we are given access to an oracle solving \textsc{End-of-Line}, because the oracle could just return the other known source. We consider the following more general problem, where we are given an \textsc{End-of-Line} graph and an explicit list of known sources.

\begin{definition}[\textsc{MS-EoL}]\label{def:mulsourceeol}
The \textsc{Multiple-Source End-of-Line} problem, abbreviated \textsc{MS-EoL}, is defined as: given circuits $S,P$ with $n$ inputs and $n$ outputs and $s_1, \dots, s_k \in \{0,1\}^n$ such that $P(s_i) = s_i \neq S(s_i)$ for all $i$, find $x \in \{0,1\}^n$ such that $P(S(x)) \neq x$ or $x \notin \{s_1, \dots, s_k\}$ such that $S(P(x)) \neq x$.
\end{definition}

In passing, let us note that in the undirected case this kind of generalisation is trivial. The undirected analogue of \textsc{End-of-Line} is \textsc{Leaf}: given an undirected graph where every vertex has degree at most $2$ and given a vertex of degree $1$, find another vertex of degree $1$, i.e.\ another leaf. Assume that we know $k$ leaves instead of just one. If $k$ is even, then the problem is not even in TFNP. If $k$ is odd, then we can add edges between known leaves until exactly one is left. Thus, the problem is equivalent to \textsc{Leaf}. This kind of reduction does not work for the directed case. Nevertheless, we obtain\footnote{This problem was discussed in an online thread (\url{https://cstheory.stackexchange.com/q/37481}). E.~Je{\v r}{\'a}bek proved membership in PPADS and PPA-$p$ for every prime $p$ (but not membership in PPAD).}:

\begin{theorem}\label{thm:mulsourceeol}
\textsc{Multiple-Source End-of-Line} is equivalent to \textsc{End-of-Line}.
\end{theorem}

\begin{remark*}[Multiple Known Sources and Sinks]
A natural generalisation of \textsc{Multiple-Source End-of-Line} is the following problem: given an \textsc{End-of-Line} graph and a list of $k$ known sources and $m$ known sinks, find another source or sink. Note that for this problem to be in TFNP, we must require $k \neq m$.
Using \cref{thm:mulsourceeol}, it is easy to see that this problem is equivalent to \textsc{End-of-Line}. If $k > m$, then we add an edge from each of the $m$ known sinks to some corresponding known source and obtain an instance with $k-m$ known sources and no known sinks. Similarly, if $k < m$, then we first reverse all directed edges and then apply the same trick.
\end{remark*}

We now give the proof of \cref{thm:mulsourceeol}. The next two sections then present additional consequences of this result.

\begin{proof}[Proof of \cref{thm:mulsourceeol}]
The reduction from \textsc{End-of-Line} to \textsc{MS-EoL} is trivial. The challenging step is the reduction from \textsc{MS-EoL} to \textsc{End-of-Line}.

Let $(S,P)$ be an instance of \textsc{MS-EoL} with a list of $k$ known sources, where the vertex set is $\{0,1\}^n$, also interpreted as $\{0, \dots, 2^n-1\}$. Without loss of generality, we assume that the known sources are $0,1,2,\dots, k-1$. This is easy to achieve by applying an efficient bijection on the vertex set.

For simplicity, we are going to assume that $P(S(z)) = z$ for all $z < k$. We also assume that for all $x$ we have $P(S(x))=x$ unless $S(x)=x$, and $S(P(x))=x$ unless $P(x)=x$. The first assumption corresponds to requiring that the first $k$ vertices indeed be sources. Note that we can check this using $O(k)$ evaluations of the circuits, i.e.\ in polynomial time in $|S|+|P|+kn$, and any ``false'' source is a solution to the \textsc{MS-EoL} instance. Here, $|S|+|P|+kn$ denotes the size of the input, i.e.\ the sum of the sizes of the two circuits and the length of the list of known sources given in the input. The second assumption corresponds to requiring that the graph is in some sense well-defined. This requirement can be enforced by a slight modification of the circuits that can be done in time polynomial in $|S|+|P|+kn$. Any solution of the modified instance is a solution of the original instance.

Let $V^s, V^t, V^p$ be the following subsets of $\{0, \dots, 2^n-1\}$:
\begin{itemize}
\item $V^s = \{x: P(x) = x, S(x) \neq x\}$. This corresponds to all the sources of the graph.
\item $V^t = \{x: S(x) = x, P(x) \neq x\}$. This corresponds to all the sinks of the graph.
\item $V^p = \{x : P(x) \neq x, S(x) \neq x\}$. This corresponds to all the vertices that are not isolated, but are neither sources nor sinks. We call those \emph{path vertices}.
\end{itemize}
Note that members of all those subsets are recognisable in polynomial time in $|S|+|P|+kn$. Let $V = V_P \cup V_S \cup V_I$. Note that isolated vertices are not contained in $V$.

Let $G=(V,E)$ be the graph represented by circuits $S,P$ (without isolated vertices). Below we will give an inductive construction of the graph $G_k=(V_k,E_k)$ that will have the following properties:
\begin{itemize}
\item The sources of $G_k$ are all the sets of the form $\{s_1, \dots, s_k\}$, where $s_1, \dots, s_k \in V^s$ are distinct sources of the original graph.
\item The sinks of $G_k$ are all the sets of the form $\{t_1, \dots, t_k\}$, where $t_1, \dots, t_k \in V^t$ are distinct sinks of the original graph.
\item Every vertex of $G_k$ can be represented using at most $kn + k^2$ bits. There exists a polynomial algorithm that for each such bit-string decides whether it represents a vertex of $G_k$ or not.
\item The successor and predecessor circuits $S_k,P_k$ have polynomial size with respect to $|S|+|P|+kn$ and can be constructed in polynomial time.
\end{itemize}

Note that the only known source of $G_k$ is $\{0, 1, \dots, k-1\}$. Any other source or any sink of $G_k$ contains at least one unknown source or sink of the original graph. Thus, if we can construct (in polynomial time in $|S|+|P|+kn$) circuits $P_k,S_k$ that represent this graph, then we have a polynomial reduction from \textsc{MS-EoL} to \textsc{End-of-Line}.

We now give a formal inductive construction of $G_\ell$. For $\ell=1$, $G$ itself already satisfies these properties (if we interpret any vertex $x$ as $\{x\}$). Let $\ell \geq 2$. Assume that we know how to construct $G_i$ for all $1 \leq i \leq \ell-1$. We then construct $G_\ell$ as follows. For any set $X$ and any $j \in \mathbb{N}$, let $\text{Subsets}(X,j) := \{A \subseteq X : |A| = j\}$, i.e.\ the set of all subsets of $X$ with cardinality exactly $j$. The set of vertices of $G_\ell$ is defined as
$$V_\ell = \text{Subsets}(V,\ell) \cup \bigcup_{i=1}^{\ell-1} \left(\text{Subsets}(V^p,i) \times V_{\ell-i}\right).$$

Let us investigate the number of bits needed to represent an element in $V_\ell$. We need $n \cdot i$ bits to represent an element in $\text{Subsets}(V^p,i)$. By induction hypothesis, $(\ell-i)n+(\ell-i)^2$ bits suffice to represent an element in $V_{\ell-i}$. Thus, for any $i \in \{1, \dots, \ell-1\}$, at most $n + (\ell-1)n+(\ell-1)^2$ bits suffice to represent an element in $V_{\ell-i}$. We add another $\lceil \log_2(\ell-1) \rceil \leq \ell-1$ bits to explicitly store the value of $i$. Thus, we can represent any element in $\bigcup_{i=1}^{\ell-1} \left(\text{Subsets}(V^p,i) \times V_{\ell-i}\right)$ using at most $\ell n+(\ell-1)^2 + \ell-1$ bits. We add one more bit to decide whether the element is in $\text{Subsets}(V,\ell)$ or not. We only need $\ell n$ bits to represent an element in $\text{Subsets}(V,\ell)$. Putting everything together, we get an upper bound on the number of bits needed to represent an element in $V_\ell$
$$1+ \max\{\ell n, \ell n+(\ell-1)^2 + \ell-1\} = \ell n + (\ell-1)^2 + \ell \leq \ell n + \ell^2.$$
Furthermore, given a bit-string, it is easy to ``decode'' it and decide whether it is a vertex of $G_\ell$ or not (and if it is not, then treat it as an isolated vertex).

We now give the construction of the predecessor and successor circuits. First, consider the case $\{x_1, \dots, x_\ell\} \in \text{Subsets}(V,\ell)$. Assume that we have reordered the elements in the set such that for some $i,j \in \{0, \dots, \ell\}$ we have $x_1, \dots, x_i \in V^p$, $x_{i+1}, \dots, x_j \in V^s$ and $x_{j+1}, \dots, x_\ell \in V^t$.
\begin{itemize}
\item If $j=\ell$ (i.e.\ only sources and path vertices), then we define 
$$S_\ell(\{x_1, \dots, x_\ell\}) = \{S(x_1), \dots, S(x_\ell)\}$$
Furthermore, if we also have $1 \leq i \leq \ell-1$ (i.e.\ at least one path vertex and source), then we define
$$P_\ell(\{x_1, \dots, x_\ell\}) = (\{x_1, \dots, x_i\}, \{x_{i+1}, \dots, x_\ell\}) \in \text{Subsets}(V^p,i) \times V_{\ell-i}$$
\item If $i = j$ and $j < \ell$ (i.e.\ only path vertices and sinks), then we define
$$P_\ell(\{x_1, \dots, x_\ell\}) = \{P(x_1), \dots, P(x_\ell)\}$$
Furthermore, if we also have $i \geq 1$ (i.e.\ at least one path vertex and sink), then we define
$$S_\ell(\{x_1, \dots, x_\ell\}) = (\{x_1, \dots, x_j\}, \{x_{j+1}, \dots, x_\ell\}) \in \text{Subsets}(V^p,j) \times V_{\ell-j}$$
(recall that we are in the case $i=j$).
\item If $i < j$ and $j < \ell$ (i.e.\ both sources and sinks, as well as path vertices potentially), then $\{x_1, \dots, x_\ell\}$ is an isolated vertex.
\end{itemize}
Now consider the case $(\{x_1, \dots, x_i\},z) \in \text{Subsets}(V^p,i) \times V_{\ell-i}$ for some $i \in \{1, \dots, \ell-1\}$.
\begin{itemize}
\item We define $P_\ell((\{x_1, \dots, x_i\},z)) = (\{x_1, \dots, x_i\},S_{\ell-i}(z))$, except if $S_{\ell-i}(z) = z$ and $P_{\ell-i}(z) \neq z$, in which case $z \in \text{Subsets}(V^t,\ell-i)$ (by induction hypothesis) and we then define $P_\ell((\{x_1, \dots, x_i\},z)) = \{x_1, \dots, x_i\} \cup z \in \text{Subsets}(V,\ell)$. Note that the two sets in this union are disjoint, because $x_1, \dots, x_i$ are path vertices of $G$, whereas $z$ only contains sinks of $G$.

\item We define $S_\ell((\{x_1, \dots, x_i\},z)) = (\{x_1, \dots, x_i\},P_{\ell-i}(z))$, except if $P_{\ell-i}(z) = z$ and $S_{\ell-i}(z) \neq z$, in which case $z \in \text{Subsets}(V^s,\ell-i)$ (by induction hypothesis) and we then define $S_\ell((\{x_1, \dots, x_i\},z)) = \{x_1, \dots, x_i\} \cup z \in \text{Subsets}(V,\ell)$.
\end{itemize}

It is straightforward to check that in the graph represented by $S_\ell,P_\ell$ every vertex has in- and out-degree at most $1$. Furthermore, by construction we also get that the sources of $G_\ell$ are exactly the vertices in $\text{Subsets}(V^s,\ell)$ and the sinks of $G_\ell$ are exactly the vertices in $\text{Subsets}(V^t,\ell)$. By induction it follows that we can construct (in polynomial time in $|S|+|P|+kn$) circuits $S_k,P_k$ that represent $G_k$.
\end{proof}

\subsection{The Imbalance problem}\label{sec:imbalance}

Up to this point, we have only considered graphs where every vertex has in- and out-degree at most $1$. However, the principle that guarantees the existence of a solution in an \textsc{End-of-Line} graph can be generalised to higher degree graphs. If we are given a directed graph and an \emph{unbalanced} vertex, i.e. a vertex with in-degree $\neq$ out-degree, then there must exist another unbalanced vertex.

Beame et al.~\cite{beame1998relative} defined the corresponding problem \textsc{Imbalance}, which is seemingly more general than \textsc{End-of-Line}. In this problem, every vertex is not constrained to have in- and out-degree at most $1$. Instead, in- and out-degree are bounded by some polynomial of the input length\footnote{Note that this trivially holds, if the input consists of circuits that explicitly output the predecessor and successor list.}. We are given a vertex that is unbalanced and have to find another unbalanced vertex (which is guaranteed to exist). The problem can be informally defined as follows:

\begin{definition}[\textsc{Imbalance}~\cite{beame1998relative}, informal]
The \textsc{Imbalance} problem is defined as: given a directed graph (represented concisely by predecessor and successor functions) and a vertex $z$ that has in-degree $\neq$ out-degree, find a vertex $x \neq z$ that also has in-degree $\neq$ out-degree.
\end{definition}

Beame et al.~\cite{beame1998relative} claim that \textsc{Imbalance} reduces to \textsc{End-of-Line}, using the same argument as for the corresponding problems on undirected graphs. However, if the graph is directed, a complication arises (that is not an issue in the undirected case). Indeed, their proof idea is incomplete, because they overlook the fact that their reduction yields an \textsc{End-of-Line} instance with \emph{multiple} known sources. Using \cref{thm:mulsourceeol} we can provide a full proof of their claim.

\begin{theorem}\label{thm:imbalanceppad}
\textsc{Imbalance} is \textup{PPAD}-complete.
\end{theorem}

Let us now define the problem formally. Similarly to \textsc{End-of-Line}, the graph is provided implicitly through circuits $S,P : \{0,1\}^n \to \{0,1\}^{d n}$ computing successors and predecessors respectively. For $x \in \{0,1\}^n$, $S(x)=(y_1,y_2, \dots, y_d) \in (\{0,1\}^n)^d$ encodes the successors of $x$ as follows: if $y_i \neq x$, then it is a potential successor of $x$. Thus, if $S(x)=(x,y,y,x,\dots,x)$, then $y$ is the only potential successor of $x$. We abuse notation and let $S(x)$ denote the set of all potential successors of $x$ (i.e.\ turn $S(x)=(y_1,y_2, \dots, y_d)$ into the set $\{y_1,y_2, \dots, y_d\} \setminus \{x\}$). We use the same interpretation and notation for the predecessor circuit $P$. Thus, a directed edge $(x,y)$ exists, if and only if $y \in S(x)$ and $x \in P(y)$.

Using this notation we can formally define the problem as follows.
\begin{definition}[\textsc{Imbalance}]
The \textsc{Imbalance} problem is defined as: given circuits $S,P : \{0,1\}^n \to \{0,1\}^{d n}$ and a vertex $z \in \{0,1\}^n$ with $|S(z)| \neq |P(z)|$, find
\begin{enumerate}
    \item $x \in \{0,1\}^n \setminus \{z\}$ such that $|S(x)| \neq |P(x)|$
    \item or $x, y \in \{0,1\}^n$ such that $y \in S(x) \land x \notin P(y)$ or $y \notin S(x) \land x \in P(y)$.
\end{enumerate}
\end{definition}

Note that the number of incoming and outgoing edges at a vertex is bounded by a polynomial in the size of the input, because $d$ is bounded by the size of the circuits. Beame et al.~\cite{beame1998relative} defined this problem in a black box model, namely the predecessor and successor functions are oracles. Since our reductions do not make use of the white box aspect of the input circuits (they don't look \emph{inside} the circuit, but just use it as a black box), the results also hold in their model.

\begin{proof}[Proof of \cref{thm:imbalanceppad}]
PPAD-hardness is trivial, since \textsc{Imbalance} generalises \textsc{End-of-Line}. To show membership in PPAD we follow the proof idea given by Beame et al.~\cite{beame1998relative}.

The undirected analogue of \textsc{End-of-Line} is \textsc{Leaf}: given an undirected graph where every vertex has degree at most $2$ and given a vertex of degree $1$, find another vertex of degree $1$, i.e. another leaf. The generalisation of \textsc{Leaf} is called \textsc{Odd}: given an undirected graph and a vertex of odd degree, find another vertex of odd degree. It is quite straightforward to show that \textsc{Odd} is equivalent to \textsc{Leaf}. Clearly, \textsc{Leaf} trivially reduces to \textsc{Odd}. To show that \textsc{Odd} reduces to \textsc{Leaf}, Papadimitriou~\cite{papadimitriou1990graph,papadimitriou1994complexity} and Beame et al.~\cite{beame1998relative} use the \emph{chessplayer algorithm} (which is based on an Euler tour argument). Intuitively, the idea is to separate vertices into multiple copies such that every copy has degree at most 2.

As claimed by Papadimitriou~\cite{papadimitriou1990graph} and Beame et al.~\cite{beame1998relative}, the chessplayer algorithm can also be applied to the directed case. Let $(S,P,z)$ be an instance of \textsc{Imbalance}, where $S,P : \{0,1\}^n \to \{0,1\}^{d n}$. First of all, note that we can assume wlog that all edges are well-defined, i.e.\ we have $y \in S(x) \Leftrightarrow x \in P(y)$. This can be achieved by a simple modification of the circuits: the modified circuit $S'$ outputs $\{y \in S(x): x \in P(y)\}$ instead of $S(x)$. Along with the analogous modification for $P$, this yields an instance where the solutions can only be of the first type, and any such solution yields a solution (of the first or second type) of the original instance.

Note that for any $x \in \{0,1\}^n$, the successor and predecessor lists can be ordered lexicographically. For any $x \in \{0,1\}^n$ and $i \geq 1$, let $S_i(x) \in \{0,1\}^n$ denote the $i$th successor of $x$, if $S(x)$ is ordered lexicographically. If $x$ has less than $i$ successors, then let $S_i(x) = x$. Finally, let $\#^{S}_x(y)$ correspond to the index of $y \in \{0,1\}^n$ in the successor list of $x \in \{0,1\}^n$, i.e.\ $S_{\#^{S}_x(y)}(x)=y$. Define $P_i(x)$ and $\#^{P}_x(y)$ analogously.

Let $\ell = \lceil \log_2 d \rceil$. We construct an \textsc{End-of-Line} instance on the vertex set $\{0,1\}^{n + \ell}$. For convenience, we use the notation $(x,i) \in \{0,1\}^n \times [2^\ell]$ to denote elements in $\{0,1\}^{n + \ell}$. The \textsc{End-of-Line} circuits $\widehat{S},\widehat{P} : \{0,1\}^{n + \ell} \to \{0,1\}^{n + \ell}$ are constructed as follows. On input $(x,i)$ the circuit $\widehat{S}$ first computes $y=S_i(x)$. If $y=x$, then it outputs $(x,i)$. Otherwise, it outputs $(y,\#^{P}_y(x))$. Similarly, on input $(x,i)$, $\widehat{P}$ computes $y=P_i(x)$ and outputs $(y,\#^{S}_y(x))$ if $y \neq x$, and $(x,i)$ otherwise. Note that $\widehat{S},\widehat{P}$ can be constructed in polynomial time in the size of $S$ and $P$.

It is easy to see that for any balanced vertex $x$ in the graph given by $(S,P)$, all of its versions $(x,\cdot)$ in the graph given by $(\widehat{S},\widehat{P})$ will either have in- and out-degree one, or be isolated. Thus, if $(x,i)$ is a source or sink, then $x$ is unbalanced in the graph given by $(S,P)$. Furthermore, all edges are well-defined.

Beame et al.\ claim that this reduction is sufficient to prove that \textsc{Imbalance} reduces to \textsc{End-of-Line} (which they call SOURCE-OR-SINK). However, consider the case where the imbalance in the known unbalanced vertex $z$ is strictly greater than one, i.e.\ $||S(z)|-|P(z)|| \geq 2$. Then, there exist $i \neq j$ such that $(z,i)$ and $(z,j)$ are both sources or both sinks. For example, if $|S(z)|=2$ and $|P(z)|=0$, then $(z,1)$ and $(z,2)$ are sources, and all the other $(z,i)$ are isolated vertices. If we consider this as an \textsc{End-of-Line} instance and pick $(z,1)$ as the known source, then $(z,2)$ is a valid solution. However, note that it does not yield a solution to the \textsc{Imbalance} instance.

Using our results on multiple-source \textsc{End-of-Line} we can complete the proof. First of all, we can make sure that $z$ is in deficiency and not in excess, i.e.\ $|S(z)|-|P(z)| \geq 1$, simply by inverting the role of $S$ and $P$. Then, the reduction described above yields an \textsc{End-of-Line} instance with known sources $(z,i+1)$, $(z,i+2), \dots,$ $(z,i+j)$, where $i=|P(z)|$ and $j=|S(z)|-|P(z)| \geq 1$. Note that all other $(z,\cdot)$ are neither sources, nor sinks. Since we can efficiently produce an explicit list of all the known sources, we obtain an instance of \textsc{MS-EoL}, which lies in PPAD by~\cref{thm:mulsourceeol}.
\end{proof}

\subsection{Looking for multiple solutions}

If we are given an \textsc{End-of-Line} instance with $k$ known sources, then we can ask for $k$ sinks or $k$ unknown sources. The problem is total, because at least $k$ sinks are guaranteed to exist. 

\begin{definition}[\textsc{$k$-EoL}]
Let $k \in \mathbb{N}$. The \textsc{$k$-Ends-of-Line} problem, abbreviated \textsc{$k$-EoL}, is defined as: given circuits $S,P$ with $n$ inputs and $n$ outputs and such that $P(z) = z \neq S(z)$ for all $z < k$, find distinct $x_1, \dots, x_k$ such that $P(S(x_i)) \neq x_i$ for all $i$ or $S(P(x_i)) \neq x_i \geq k$ for all $i$.
\end{definition}
Intuitively, this problem seems harder than \textsc{End-of-Line} or \textsc{MS-Eol}, because we are now looking for more than one solution. However, using \cref{thm:mulsourceeol} we can show:

\begin{theorem}\label{thm:eol-mult-sol}
For any $k \in \mathbb{N}$, \textsc{$k$-Ends-of-Line} is \textup{PPAD}-complete.
\end{theorem}

\begin{proof}
The non-trivial direction is the reduction to \textsc{End-of-Line}. Buss and Johnson~\cite{buss2012propositional} have shown that PPAD, PPADS, PPA and PLS are closed under Turing reductions, by providing a way to transform a Turing reduction (to a complete problem of each of those classes) into a many-one reduction. Thus, it suffices to provide a Turing reduction from \textsc{$k$-Ends-of-Line} to \textsc{End-of-Line}. By \cref{thm:mulsourceeol} we can efficiently simulate an oracle for \textsc{MS-EoL}, using an oracle for \textsc{End-of-Line}. We can solve an instance of \textsc{$k$-EoL} by making repeated calls to \textsc{MS-EoL} oracles, with a list of all the currently known sources. If the oracle call returns a new source, then we add it to our list. If the oracle returns a sink, then we add an edge from this sink to one of the known sources, and remove that source from the list. It is easy to see that after at most $2k$ oracle calls we will have obtained at least $k$ sinks or at least $k$ new sources. Note that the list of known sources will have length at most $2k$.
\end{proof}

\begin{remark*}
Note that the proof still works if we are given an explicit list of the known sources in the input (as in the definition of \textsc{MS-EOL}), i.e.\ $k$ does not have to be fixed. Furthermore, the same proof also yields PPAD-completeness for the following problem. Fix some polynomial $p$. The problem is: given $k$ sources, find $k$ sinks or $p(k)$ sources. This seems quite surprising, as one might have expected this problem to be closer to PPADS.
\end{remark*}

We close this section by giving some analogous results for the class PPADS and its canonical complete problem \textsc{Sink}~\cite{papadimitriou1990graph,beame1998relative}. \textsc{Sink} is identical to \textsc{End-of-Line}, except that we only accept a sink as a solution and are not interested in other sources. In this case the results are easier to obtain, because there is no need for an analogue of \cref{thm:mulsourceeol}.

\begin{definition}[\textsc{Sink}~\cite{papadimitriou1990graph,beame1998relative}]
The \textsc{Sink} problem is defined as: given circuits $S,P$ with $n$ inputs and $n$ outputs and such that $P(0) = 0 \neq S(0)$, find $x$ such that $P(S(x)) \neq x$.
\end{definition}

Unlike \textsc{End-of-Line}, it is easy to prove that multiple source \textsc{Sink} is equivalent to \textsc{Sink}. Consider the problem \textsc{MS-Sink}, where we are given a graph and a list of known sources and are looking to find a sink. It is easy to see that this problem is equivalent to \textsc{Sink}. A reduction from \textsc{Sink} to \textsc{MS-Sink} is given by simply taking $k$ copies of the original \textsc{Sink} instance graph. The reduction in the other direction is even more trivial. Indeed, we can just ignore the extra $k-1$ sources we know, because we are only interested in sinks.

We can define the analogous problem to \textsc{$k$-EOL}, where we look for multiple sinks.
\begin{definition}[\textsc{$k$-Sinks}]
Let $k \in \mathbb{N}$. The \textsc{$k$-Sinks} problem is defined as: given circuits $S,P$ with $n$ inputs and $n$ outputs and such that $P(z) = z \neq S(z)$ for all $z < k$, find distinct $x_1, \dots, x_k$ such that $P(S(x_i)) \neq x_i$ for all $i$.
\end{definition}

\begin{theorem}
For any $k \in \mathbb{N}$, \textsc{$k$-Sinks} is \textup{PPADS}-complete.
\end{theorem}

\begin{proof}
\textsc{Sink} easily reduces to \textsc{$k$-Sinks} by taking $k$ copies of the graph. The other direction can again be proved by using the result by Buss and Johnson~\cite{buss2012propositional}. A Turing reduction from \textsc{$k$-Sinks} to \textsc{Sink} is obtained by doing the following: given an instance of \textsc{$k$-Sinks}, use the oracle to solve the \textsc{Sink} problem on this instance, then add an edge from the sink we just obtained to one of the known sources. Doing this $k$ times yields $k$ distinct sinks of the original instance.
\end{proof}

\begin{remark*}
Just like \cref{thm:eol-mult-sol}, this result also holds if there is a polynomial number of sources, in particular if they are given explicitly in the input.
\end{remark*}

\subsection*{Acknowledgements}

This work was supported by an EPSRC doctoral studentship (Reference 1892947).

\bibliography{hairyballthm}
\bibliographystyle{abbrv}

\clearpage

\appendix

\section{Reduction from 2D-Hairy-Cube to 2D-Hairy-Ball}\label{sec:app:cubetoHB}

In this section we prove the following result.
\begin{proposition}\label{prop:cubetoHB}
\textsc{2D-Hairy-Cube} reduces to \textsc{2D-Hairy-Ball}.
\end{proposition}

The main idea of the proof is simple: first we obtain a continuous field on the surface of the cube by interpolation. Then, we transfer the field onto the sphere. Various technical difficulties arise, but they can all be overcome. One difficulty is that in order to simulate the Boolean circuit describing the discrete field, the arithmetic circuit first has to do ``bit extraction'', but this cannot be done exactly. This issue can be resolved by using a standard technique called the \emph{averaging trick}.

\begin{proof}[Proof of \cref{prop:cubetoHB}]
The reduction works by constructing an arithmetic circuit that will (among other things) simulate the Boolean circuit. To do this we need to do ``bit extraction'' on a rational input of the arithmetic circuit, i.e.\ compute the $N$ most significant bits in its binary representation. However, this is impossible to do exactly using an arithmetic circuit with gates $\{+$, $\times \zeta$, $\max\}$, because this kind of circuit can only compute continuous functions. As a result, any bit extraction method will make mistakes in some small regions. To circumvent this obstacle, we will simulate the evaluation of the Boolean circuit at multiple points close to the actual input. By choosing those points wisely we can ensure that at most three of them can result in a bogus bit extraction. Thus, if we output the average value obtained from all those points, the mistakes will have a very limited impact on the final result. This averaging technique is a standard tool in the field, see~\cite{daskalakis2009complexity,CDT}.

Let $C: \{0,1\}^n \times \{0,1\}^n \times \{0,1\}^3 \to \{0,1\}^2$ be an instance of \textsc{2D-Hairy-Cube}. We interpret the cube as being the unit cube $\{y \in \mathbb{R}^3 : \|y\|_\infty = 1\}$. In particular, every squarelet has side-length $2^{-n+1}$. The first step of the reduction is to perform a series of transformations on the input circuit $C$. First of all, we modify the instance by subdividing every squarelet on the surface of the cube into four small squarelets (of side-length $2^{-n}$). Any small squarelet that is not adjacent to an edge of the cube is assigned the same arrow as the original squarelet containing it. Any small squarelet that is adjacent to an edge of the cube is assigned the average of the arrow in its original squarelet and the original squarelet on the other side of the edge. If a small squarelet is adjacent to two edges of the cube (i.e.\ it lies in the corner of a face), then it is assigned the average of the arrows in the three original squarelets adjacent to this corner. From now on, when we talk about squarelets we mean the small squarelets, which have now replaced the original grid.

Next, we move from squarelets to cubelets. We place cubelets of side-length $2^{-n}$ around the surface of the unit cube, both outside and inside. Namely, every squarelet on the surface of the cube yields two cubelets: one on each side of the face containing the squarelet. Thus, the squarelet is a common face of its two cubelets. The two cubelets corresponding to a squarelet are assigned the same arrow as the squarelet. Note that this is well-defined, since for any cubelet that has multiple faces that lie on the unit cube, all the corresponding squarelets contain the same arrow by construction. We also add cubelets on the outside of the unit cube along its edges, i.e.~such that one edge of the cubelet is part of an edge of the unit cube. Such a cubelet is assigned the arrow of the squarelets it has a common edge with. Finally, for every corner of the unit cube, we add a cubelet outside the unit cube, such that the cubelet has a corner in common with the unit cube. This cubelet is assigned the same arrow as the three corresponding corner squarelets. Note that again all the cubelets are assigned an arrow in a unique and well-defined way. Furthermore, any point $y \in \mathbb{R}^3$ with $\|y\|_\infty \in [1-2^{-n},1+2^{-n}]$ lies in a cubelet (it might lie on the boundary between multiple cubelets).

Following the description above we can construct the circuit $C': [2^{n+1}+2] \times [2^{n+1}+2] \times [2^{n+1}+2] \to \{0,1\}^{\ell}$. The input in $[2^{n+1}+2] \times [2^{n+1}+2] \times [2^{n+1}+2]$ describes a cubelet in the subdivision of the cube $\{y \in \mathbb{R}^3: \|y\|_\infty \leq 1+2^{-n}\}$ into cubelets of side-length $2^{-n}$. Note that the input is actually of the type $\{0,1\}^{n+2}$, but we interpret it as an integer in $[2^{n+2}]$ (and we are only interested in inputs in $[2^{n+1}+2]$). For cubelets that do not correspond to the ones described above, the circuit can have an arbitrary output. For the cubelets we described above, the circuit outputs the corresponding arrow. Note that there are $6$ cardinal arrows overall. The circuit has $\ell=6+6^2+6^3$ output bits and for every valid cubelet input, exactly one of those output bits will be set to $1$. Note that $\ell$ is large enough to ensure that we can also output the average of two or three cardinal arrows. This circuit can be constructed in polynomial time from $C$.

Now consider a point $y$ on the surface of the unit cube and let $e=(1,1,1)$. Let
$$S_y = \{y + (i-50) 2^{-n-8} \cdot e : i=1, \dots, 100\}$$
Note that since $y$ lies on the surface of the cube, all the points in $S_y$ lie in at least one of the cubelets we defined above. Furthermore, note that any two cubelets that each contain a point in $S_y$, have to be adjacent, i.e.\ they have a non-empty intersection.

It is easy to see that the kind of arithmetic circuit we consider here can simulate any single gate of a Boolean circuit and thus the whole circuit $C'$. However, it needs to somehow obtain the input bits. Since any output (and intermediate output) of the arithmetic circuit depends continuously on the inputs, it is impossible to construct a circuit that extracts bits correctly for all inputs. Thus, we will have to take into account the fact that on some inputs, the bit extraction will not output $0$ or $1$. In this case, the simulated Boolean circuit will have some ``bogus'' output. In particular, there might not be exactly one bit in the output that is set to $1$. However, we can ensure that the output gates of the simulated circuit always output some number in $[0,1]$ by adding adequate $\max$ and $\min$ gates at the end. Furthermore, we can also ensure that those numbers add up to something that it at most $1$, which ensures that even the bogus outputs will yield a vector that has $\ell_1$-norm at most $1$. This is done by computing the sum $S$ of all ``fake'' output bits and then subtracting $S - 1$ from all fake output bits (and then taking the $\max$ with $0$). It is easy to see that this ensures that the sum of outputs is at most $1$, while also not changing the output if it is valid (i.e.\ exactly one bit set to $1$).

The simulated Boolean circuit will output $\ell$ values, corresponding to the $\ell$ output bits of $C'$. As mentioned above, we can ensure that these values lie in $[0,1]$ and sum up to at most $1$. The output vector is then computed by taking the weighted average of the $\ell$ possible arrows, where the weights are given by the $\ell$ output values. This is achieved by hard-coding every coordinate of every arrow as a constant in the arithmetic circuit. The cardinal arrows are taken to be the unit vectors in the corresponding direction. The other arrows are taken to be all the averages of two or three of the cardinal vectors. Note that if exactly one output value is $1$ (which will be the case whenever $C'$ is simulated with correct input bits), then the output will be the correct vector that $C'$ outputs. Also note that, in any case, the output vector will have $\ell_1$-norm at most $1$.

We now describe how to perform bit extraction using an arithmetic circuit. Let $t \in [0,1]$. The arithmetic circuit can compute
$$b = \min(1, 2^{-n-10} \cdot \max (0,t-1/2))$$
using a number of gates that is polynomial in $n$. If $t \leq 1/2$, we obtain $b=0$. If $t \geq 1/2 + 2^{n+10}$, then we obtain $b=1$. For intermediate values of $t$, we obtain a number $b$ that is not a bit. Performing the same operation on $t'=2t-b$ will extract the second bit, etc. Consider any $z \in \mathbb{R}^3$ with $\|z\|_\infty \leq 1+2^{-n}$. If we extract the first $n+2$ bits of
$$(z_i+1+2^{-n})\frac{(2^{n+1}+2)2^{-n-2}}{2+2^{-n+1}}$$
for $i=1,2,3$ then we will obtain the input bits we need for $C'$. In particular, if $z$ lies in one of our cubelets and is at ($\ell_\infty$-norm) distance at least $2^{-n-10}$ from the boundary of the cubelet, then this extraction correctly recovers the cubelet in which it lies. From the construction of $S_y$ above, it follows that at most $3$ points in $S_y$ can suffer from incorrect extraction. Thus, for at least $97$ points we will correctly determine the cubelet in which they lie. The bit extraction above can be done using an arithmetic circuit of polynomial size in $n$. In particular, note that all constants that we use have a number of bits that can be bounded by a polynomial in $n$.

We are now ready to describe the arithmetic circuit $F$ with three inputs and outputs, using gates $\{+, \times \zeta, \max \}$, that will be the output of the reduction. Let $x = (x_1,x_2,x_3)$ be the (rational) input to the circuit. The circuit $F$ does the following:
\begin{enumerate}
    \item Project $x$ onto the surface of the $1/2$-cube (i.e.\ the set of points $u$ such that $\|u\|_\infty = 1/2$) and then map it to the surface of the unit cube. Namely, compute $y=(y_1,y_2,y_3)$ where $y_i = 2 \times \min(1/2,\max(-1/2,x_i))$.
    \item Compute the coordinates of the $100$ points in $S_y$.
    \item For each point in $S_y$ : perform the bit extraction described above, then simulate $C'$ with those bits as input. (Ensure output vector has $\ell_1$-norm at most $1$ using trick described above.)
    \item Output the average of all $100$ output vectors.
\end{enumerate}

We now show that there exists some constant $\varepsilon > 0$ such that any solution to the \textsc{2D-Hairy-Ball} instance $(F,\varepsilon)$ yields a solution to the discrete instance. For now we assume that the bit extraction is always correct and thus the vector output for each point in $S_y$ is computed correctly by the simulated circuit. Let $v$ denote the average of all these vectors. This is what the arithmetic circuit outputs at the end. We start with some key observations. First of all, every cubelet obtained its arrow from some squarelets in the original instance. A cubelet lying next to a face of the cube (but not right next to an edge), obtained its arrow from the squarelet containing its intersection with the face of the cube. A cubelet lying at an edge of the cube, obtained its arrow by taking the average of the arrows of two squarelets lying on either side of the edge. Finally, a corner cubelet obtained its arrow by taking the average of the arrows of the three corner squarelets. The key observation here is that if two cubelets are adjacent (i.e.\ have at least one point in common), then the corresponding original squarelets must all be pairwise adjacent. The second observation is that all the cubelets touched by our sample $S_y$ are adjacent. When we take the average $v$ of all outputs over $S_y$, this actually corresponds to taking some weighted average over all cardinal vectors present in the corresponding original squarelets. If there is no nearby solution in the original discrete instance, then all these cardinal vectors have to lie in the same orthant of $\mathbb{R}^3$. This is easy to see, because if they don't lie in the same orthant, then two of them have to be opposite, which yields a solution (since all the corresponding original squarelets are adjacent). From this it follows that the weighted average $v$ of all these cardinal vectors will have $\ell_1$-norm equal to $1$, if there is no nearby solution in the original discrete instance.

The next step is to show that if there is no discrete solution close to the point $y$ corresponding to $x$, then $\|P_x[v]\|_\infty$ is lower-bounded by a constant. Recall that $P_x$ denotes the projection onto the tangent space to the unit sphere at $x$. We consider three cases depending on the type of cubelet involved. A \emph{face cubelet} is one that only has one corresponding original squarelet. An \emph{edge cubelet} has two corresponding original squarelets and a \emph{corner cubelet} has three corresponding original squarelets.

Consider any $x \in S^2$ such that there is no discrete solution close to $y$ (namely: no solution involving squarelets that are adjacent to a squarelet containing $y$).
\begin{itemize}
\item Case 1: $x$ is such that all points in $S_y$ lie in face cubelets.

Note that in this case, all the corresponding original squarelets lie on the same face. As a result, the average $v$ will lie in the plane parallel to this face and going through the origin. Furthermore, note that $y$ cannot lie too close to any edge of the cube, otherwise we would get edge cubelets. Assume wlog that $y$ lies on the face $\{z_3=1\}$. It follows that $x$ must be such that $x_1 \in (-1/2,1/2), x_2 \in (-1/2,1/2), x_3 > 1/\sqrt{2}$. As stated above we have $\|v\|_1=1$. One can check that in this case a (very crude) bound is $\|P_x[v]\|_\infty \geq \frac{1}{3 \sqrt{2}}$.

\item Case 2: $x$ is such that at least one point in $S_y$ lies in an edge cubelet and no point lies in a corner cubelet

Note that there is a special case where $S_y$ can have points contained in edge cubelets from two different edges. This special case is handled exactly as in case 3 below. Here we only consider the case where all edge cubelets come from the same edge denoted $E$. Assume wlog that $E$ is the edge given by $\{z_1=1, z_2=-1\}$. $v=(v_1,v_2,v_3)$ is the weighted average of the cardinal vectors contained in all the corresponding original squarelets. All those squarelets must be adjacent to $E$ on either side (and they are all adjacent to each other as noted earlier). Observe that $(v_1,v_2)$ cannot lie in the interior of the two quadrants given by $\{v_1<0,v_2>0\}$ and $\{v_1>0,v_2<0\}$. Otherwise, we would have two cardinal vectors forming an edge solution.

Since $y$ is at distance at most $2^{-n+1}$ from $E$, it follows that $x$ satisfies $x_1 \geq 1/2-2^{-n}, x_2 \leq -1/2+2^{-n}, x_3 \in (-1/2,1/2)$. One can check that for $n \geq 4$, we get that the angle $\alpha$ between $x$ and $v$ must satisfy $\cos \alpha \leq \sqrt{207}/16$. It follows that
\begin{equation*}
\begin{split}
\|P_x[v]\|_\infty \geq \frac{1}{\sqrt{3}} \|P_x[v]\|_2 \geq \frac{1}{\sqrt{3}} (\|v\|_2 - |\langle v,x \rangle| \|x\|_2) &\geq \frac{1}{\sqrt{3}} \|v\|_2 (1 - \sqrt{207}/16)\\
&\geq (1 - \sqrt{207}/16)/3 > 0.033
\end{split}
\end{equation*}

\item Case 3: $x$ is such that at least one point in $S_y$ lies in a corner cubelet

Assume wlog that the corner is the one given by $\{z_1=1,z_2=-1,z_3=1\}$. $v$ is the weighted average of the cardinal vectors contained in all the corresponding original squarelets. It is easy to see that the original squarelets concerned are exactly the original corner squarelets. By inspection one can check that $v$ cannot lie in the two orthants given by $\{v_1 \leq 0, v_2 \geq 0, v_3 \leq 0\}$ and $\{v_1 \geq 0, v_2 \leq 0, v_3 \geq 0\}$. Since $y$ is at distance at most $2^{-n+1}$ from the corner, it follows that $x$ satisfies $x_1 \geq 1/2-2^{-n}, x_2 \leq -1/2+2^{-n}, x_3 \geq 1/2 + 2^{-n}$. Similarly to the previous case, we again obtain that the angle $\alpha$ between $x$ and $v$ must satisfy $\cos \alpha \leq \sqrt{207}/16$ (for $n \geq 4$). Thus, the same bound holds for this case too.
\end{itemize}

Unfortunately, all output vectors for points in $S_y$ might not be computed correctly. Fortunately, at most $3$ of them might yield bogus outputs. Thus we can write $v=\frac{97}{100} v_g + \frac{3}{100} v_b$, where $v_g$ and $v_b$ are the averages of the good and bad outputs, respectively (assume wlog that there are exactly 3 bad outputs). The arguments in the previous paragraphs yield $\|P_x[v_g]\|_\infty \geq 0.033$. By construction of the arithmetic circuit $F$, we have $\|v_b\|_1 \leq 1$, which implies $\|P_x[v_b]\|_2 \leq 1$, and thus $\|P_x[v_b]\|_\infty \leq 1$. Putting everything together we get
\begin{equation*}
\begin{split}
\|P_x[F(x)]\|_\infty = \left\|P_x\left[\frac{97}{100} v_g\right] + P_x\left[\frac{3}{100} v_b\right]\right\|_\infty &\geq \left\|P_x\left[\frac{97}{100} v_g\right]\right\|_\infty - \left\|P_x\left[\frac{3}{100} v_b\right]\right\|_\infty\\
&\geq \frac{97}{100}0.033 - \frac{3}{100} = 0.00201.
\end{split}
\end{equation*}
Thus picking $\varepsilon = 0.002$ yields the result.
\end{proof}

\section{Reduction from \texorpdfstring{$k$}{k}D-Hairy-Ball to \texorpdfstring{$(k+2)$}{(k+2)}D-Hairy-Ball}\label{sec:app:kHBtok+2HB}

Since \textsc{2D-Hairy-Cube} is PPAD-hard (\cref{prop:cubePPADhard}) and it reduces to \textsc{2D-Hairy-Ball} (\cref{prop:cubetoHB}), it follows that \textsc{2D-Hairy-Ball} is also PPAD-hard. In order to prove \cref{thm:kHB-PPAD-hard}, we reduce \textsc{$k$D-Hairy-Ball} to \textsc{$(k+2)$D-Hairy-Ball} for all even $k \geq 2$.

\begin{proposition}
For all even $k \geq 2$, \textsc{$k$D-Hairy-Ball} reduces to \textsc{$(k+2)$D-Hairy-Ball}.
\end{proposition}

If all computations could be performed exactly with infinite precision, then we would just use the following idea. Assume that we are given $f : S^k \to \mathbb{R}^{k+1}$ that is a Lipschitz-continuous tangent vector field on $S^k$. We define $g : S^{k+2} \to \mathbb{R}^{k+3}$ as follows. For $x=(x_1,x_2, \dots, x_{k+3}) \in S^{k+2}$, let
$$g(x) = \left(\|(x_1,\dots,x_{k+1})\|_2 \cdot f \left( \frac{(x_1,\dots,x_{k+1})}{\|(x_1,\dots,x_{k+1})\|_2} \right), -x_{k+3}, x_{k+2}\right)$$
and $g(x) = (0,\dots,0,-x_{k+3},x_{k+2})$ if $(x_1, \dots, x_{k+1}) = 0$.
It is easy to check that $g$ is a Lipschitz-continuous tangent vector field on $S^{k+2}$. Furthermore, given $x \in S^{k+2}$ such that $\|g(x)\|_\infty \leq \varepsilon / \sqrt{2}$, it holds that $y=(x_1, \dots, x_{k+1})/\|(x_1, \dots, x_{k+1})\|_2$ is a well-defined point on $S^k$ and satisfies $\|f(y)\|_\infty \leq \varepsilon$.

\begin{proof}
Since we cannot perform all these computations exactly, the reduction we construct is slightly more cumbersome. Assume that we are given an instance $(F, \varepsilon)$ of \textsc{$k$D-Hairy-Ball}. We construct an arithmetic circuit $G$ with $k+3$ inputs and outputs, using gates $\{+, \times \zeta, \max\}$ as follows. From the proof of \cref{lem:linearHBlipschitz} we know that $F$ is $L$-Lipschitz and $\|F(x)\|_\infty \leq L$ for $\|x\|_\infty \leq 1$, where $L = k 2^{\text{size}(F)^2 + 3}$. Let $(x_1, \dots, x_{k+3})$ denote the input. Pick $\ell = \lceil 1/2 \log_2(32kL/\varepsilon) \rceil$. The circuit $G$ performs the following steps:
\begin{enumerate}
    \item Compute $N := \max\{|x_{k+2}|, |x_{k+3}|\}$
    \item Compute $v := F(x_1, \dots, x_{k+1})$
    \item Compute $b := \max\{0, 2L-2^\ell \cdot L \cdot N\}$
    \item Output $w := (\min\{b,\max\{-b,v\}\}, -x_{k+3}, x_{k+2})$
\end{enumerate}
In the last expression, the $\min$ and $\max$ operations are applied component-wise on the vector $v$. Let $\varepsilon' = \min \{2^{-\ell-1}, \varepsilon/4\}$.
$(G, \varepsilon')$ is an instance of \textsc{$(k+2)$D-Hairy-Ball}. Let $x \in S^{k+2}$ be a solution, i.e.\ such that $\|P_x[G(x)]\|_\infty \leq \varepsilon'$.

The two last coordinates of $P_x[G(x)]$ are $x_{k+2} + \alpha x_{k+3}$ and $x_{k+3} - \alpha x_{k+2}$ respectively, where $\alpha = \langle w,x \rangle$. Since $\|P_x[G(x)]\|_\infty \leq \varepsilon'$, it follows that $|x_{k+2} + \alpha x_{k+3}| \leq \varepsilon'$ and $|x_{k+3} - \alpha x_{k+2}| \leq \varepsilon'$. From this we get $|x_{k+2} + \alpha^2 x_{k+2}| \leq |x_{k+2} + \alpha x_{k+3}| + |\alpha||-x_{k+3} + \alpha x_{k+2}| \leq (1+|\alpha|) \varepsilon'$. It follows that $|x_{k+2}| \leq \varepsilon' (1+|\alpha|) / (1+\alpha^2) \leq 2 \varepsilon' \leq 2^{-\ell}$. Similarly we also get $|x_{k+3}| \leq 2^{-\ell}$.

Let $x' = (x_1, \dots, x_{k+1})$ and $z = x'/\|x'\|_2 \in S^k$. We have $\|x'\|_2 \geq \sqrt{1-2^{1-2\ell}}$ and thus $\|x'-z\|_2 \leq 1 - \sqrt{1-2^{1-2\ell}} \leq 2^{1-2\ell} \leq \varepsilon/16kL$. It follows that $\|F(x')-F(z)\|_\infty \leq \varepsilon/16k$. From $N \leq 2^{-\ell}$ we get $b \geq L$, which implies $w = (F(x'), -x_{k+3}, x_{k+2})$.

Note that $\langle w, x \rangle = \langle F(x'), x' \rangle$. Thus, the first $k+1$ coordinates of $P_x[G(x)]$ are $u := F(x') - \langle F(x'), x' \rangle x'$. Since $x$ is a solution we have $\|u\|_\infty \leq \varepsilon' \leq \varepsilon/4$. We also have
\begin{equation*}\begin{split}
\|u - P_z[F(z)]\|_\infty &= \|F(x') - \langle F(x'), x' \rangle x' - (F(z) - \langle F(z), z \rangle z)\|_\infty\\
&\leq \|F(x')-F(z)\|_\infty + |\langle F(z) - \|x'\|_2^2 F(x'), z \rangle| \|z\|_2\\
&\leq \varepsilon/16k + \|F(z) - \|x'\|_2^2 F(x')\|_2\\
&\leq (1+\sqrt{k+1}) \varepsilon/16k + (1 - \|x'\|_2^2) \sqrt{k+1} L\\
&\leq (1+\sqrt{k+1})\varepsilon/16k + \sqrt{k+1} 2^{1-2\ell} L\\
&\leq (1+2\sqrt{k+1}) \varepsilon/16k \leq \varepsilon/4
\end{split}\end{equation*}

Thus, it follows that $\|P_z[F(z)]\|_\infty \leq \varepsilon/2$. Compute $z' \in S^k$ such that $\|z - z'\|_\infty \leq \varepsilon/2L$ (e.g.\ by using the stereographic projection technique). By \cref{lem:linearHBlipschitz} it follows that $\|P_{z'}[F(z')]\|_\infty \leq \varepsilon$.
\end{proof}

\section{Hairy-\texorpdfstring{$g$}{g}-Torus is PPAD-hard}\label{sec:app:gTorus-ppad-hard}

\begin{proposition}\label{prop:gTorus-PPAD-hard}
For any $g \geq 2$, \textsc{Hairy-$g$-Torus} is \textup{PPAD}-hard.
\end{proposition}

\begin{proof}

We reduce from the PPAD-complete problem \textsc{2D-Variant-Brouwer} (\cref{sec:brouwer-ppad-hard}). Let $C: [2^n] \times [2^n] \to \{(1,0),(-1,0),(0,1),(0,-1)\}$ be an instance of \textsc{2D-Variant-Brouwer}. Note that we can slightly change the boundary conditions in the definition of the problem without changing its complexity. Specifically, as long as opposite sides of the square have opposite arrows, the problem remains the same. In this proof we use \textsc{2D-Variant-Brouwer} instances where the top boundary contains right-arrows and the right boundary contains up-arrows (and the rest of the boundary is diametrically opposite). The boundary conditions are chosen this way to be compatible with the base field $F_B$.

We give the proof for $g=2$. The same arguments work for $g \geq 3$. We are going to pick an $\varepsilon > 0$ and construct an arithmetic circuit $G$ that computes a field on the $2$-torus, such that any $\varepsilon$-approximate zero of the projection of $G$ yields a solution to the \textsc{2D-Variant-Brouwer} instance. $G$ will compute a close approximation of the base field $F_B$ (defined in \cref{sec:gTorus}) and we will embed an interpolation of the \textsc{2D-Variant-Brouwer} instance in each of the two flat regions. Thus, any approximate zero will have to lie in one of the regions corresponding to the embedded Brouwer instance.

We now describe how to construct $G$. In the two embedding regions, we will use the averaging trick to ensure that $G(x)$ is small only if a \textsc{2D-Variant-Brouwer} solution is nearby. Since this has been explained in detail in the proof of \cref{prop:cubetoHB} (\cref{sec:app:cubetoHB}), we now treat this as a subroutine that we know how to implement efficiently. If one of the sampling points lies outside of the unit square, then we naturally extend the Brouwer instance there. This can be done by adding an extra layer of squarelets around the unit square and assigning the same arrow as the closest squarelet. Here we use the averaging trick with equiangle sampling of $8$ points. In this case, we get that $\|G(x)\|_\infty \geq 1/4$ if no solution is nearby. Thus, as long as $\varepsilon < 1/4$, any $\varepsilon$-approximate zero yields a \textsc{2D-Variant-Brouwer} solution.

Given input $x=(x_1,x_2,x_3)$ the arithmetic circuit $G$ does the following:
\begin{enumerate}
    \item Compute $v_L=(x_2,-x_1,0)$, $v_R=(x_2,-x_1+5,0)$ and $v_M=(2x_2,0,0)$
    \item Use the averaging trick on $C$ with input $(x_1-2,x_2+1/2)$. Obtain average $v_C$
    \item Compute indicators $b_L, b_R, b_C$ and $b_M$ (see below)
    \item Output $G(x) = h(v_L,b_L) + h(v_R,b_R) + h(v_M,b_M) + h(v_C,b_C)$
\end{enumerate}
where $h(v,b)$ outputs $\max\{-b,\min\{b,v_i\}\}$ for each of the three coordinates of $v$. The $b$-variables are indicator functions for various regions of the $2$-torus. Since the arithmetic circuit can only compute continuous functions, we have to use \emph{approximate} indicator functions, i.e.\ their value goes from 0 to 1 continuously in some small region close to the boundary of the region. The indicator bits are computed as follows:
\begin{equation*}\begin{split}
b_L(x_1,x_2) &:= \max\{0,1+2^\ell\min\{0,-x_1\},1+2^\ell\min\{0,2-x_1,2-x_2,x_2+2\}\}\\
b_R(x_1,x_2) &:= b_L(5-x_1,x_2)\\
b_C(x_1,x_2) &:= \max\{0,2^\ell\min\{2^{-\ell},x_1-2,3-x_1,x_2+1/2,1/2-x_2\}\}\\
b_M(x_1,x_2) &:= 1 - b_L(x_1,x_2) - b_R(x_1,x_2) - b_C(x_1,x_2)
\end{split}\end{equation*}

Note that by construction, for any $(x_1,x_2,x_3) \in T_2$ we have $b_L+b_R+b_M+b_C = 1$. Using $\|h(v,b)\|_\infty \leq b$, it follows that $\|G(x)\|_\infty \leq 1$. Moreover, $G(x)$ always lies in the $(x_1,x_2)$-plane (i.e.\ $[G(x)]_3 = 0$). Note also that $h(v,b)$ lies in the same quadrant as $v$ and $\|h(v,b)\|_p \geq \min\{b, \|v\|_p\}$ for $p \in \{1,2, \infty\}$.

Consider any $x=(x_1,x_2,x_3) \in T_2$ that does not lie in any of the two embedding regions, i.e.\ $x$ is in $L$, $R$ or $M$. If $x$ lies in the $L$ (or $R$) region, then $\|P_x^\varepsilon[G(x)]\|_2=\|G(x)\|_2=\|h(v_L,b_L)\|_2 \geq 1$. If $x$ lies in $M$, but away from the boundary with $L$ or $R$, then we have $\|P_x^\varepsilon[G(x)]\|_2=\|G(x)\|_2 = \|h(v_M,b_M)\|_2\geq 1/2$. If $x$ lies close to the boundary between $M$ and $L$ (or $R$), then $G(x)$ is the sum of two vectors that have length at least $\min\{b_L,1\}=1-b_M$ and $\min\{b_M,1/2\}$ respectively. Furthermore, it is easy to check that these vectors lie in the same quadrant. It follows that $\|G(x)\|_2 \geq 1/2$. Using $\|G(x)\|_2 \leq \sqrt{2}$, one can check that in this case we must have $\|G(x) - P_x[G(x)]\|_2 \leq 2^{-\ell+1}$. By choosing $\ell \geq 3$, we can ensure that $\|P_x[G(x)]\|_2 \geq 1/4$, i.e.\ $\|P_x[G(x)]\|_\infty \geq 1/(4 \sqrt{3})$. Thus, if $\varepsilon \leq 1/8$, then we must have $\|P_x^\varepsilon[G(x)]\|_\infty \geq  1/16$. Thus, if $\varepsilon < 1/16$ no solution can occur anywhere, except in the embedding region.

Now consider any $x=(x_1,x_2,x_3) \in T_2$ that lies in one of the two embedding regions. Note that $G(x)=P_x[G(x)]=P_x^\varepsilon[G(x)]$ here. Away from the boundary (i.e.\ $2^{-\ell}$ away from it), the averaging trick ensures that a solution can only occur if (at least) two of the averaging points lie in squarelets with opposite arrows (assuming we have chosen $\varepsilon < 1/4$). It remains to figure out what happens close to the boundary. Pick $\ell = n+1$. Consider a point that lies in the upper left part of the boundary, i.e.\ $x_1 \in [2,2+2^{-n-1}], x_2 \in [0,1/2]$ or $x_1 \in [2,3/2], x_2 \in [1/2-2^{-n-1},1/2]$. For now assume that there are no bogus points in the average trick. Then, $v_L$, $v_M$ and $v_C$ all lie in the same quadrant ($v_1 \geq 0, v_2 \leq 0$). Furthermore, we have $\|v_L\|_1 \geq 2$, $\|v_M\|_1 \geq 1/2$ and $\|v_C\|_1 = 1$. Thus, $\|G(x)\|_1 \geq  \min\{b_L, 2\} + \min\{b_M, 1/2\} + \min\{b_C, 1\} \geq 1/2$. There are at most $2$ bogus points out of $8$ samples and they all yield vectors with $1$-norm at most $1$. Thus, taking into account bogus points, we get $\|G(x)\|_1 \geq 1/2-2/8=1/4$, which implies $\|G(x)\|_\infty \geq 1/8$. The same argument applies to the other parts of the boundary.

Thus, by picking $\ell = n+1$ and $\varepsilon = 1/32$, from any $x \in T_2$ with $\|P_x^\varepsilon[G(x)]\|_\infty \leq \varepsilon$, we can efficiently find a solution to the \textsc{2D-Variant-Brouwer} instance.
\end{proof}

\end{document}